%% file: main.tex
\documentclass{article}
\usepackage[margin=1in]{geometry}

\usepackage{amsmath}
\usepackage{amsthm}
\usepackage{amsfonts}
\usepackage{booktabs} 

\usepackage[ruled]{algorithm2e}
\usepackage{thmtools}
\usepackage{thm-restate}
\usepackage{subcaption}
\usepackage{bm}
\usepackage{hyperref}

\usepackage{cleveref}
\usepackage{mathtools}
\usepackage{nomencl}
\usepackage{xcolor,pifont}

\usepackage{mathrsfs}
\usepackage{etoolbox}

\makenomenclature
\renewcommand\nomgroup[1]{%
	\item[\bfseries
	\ifstrequal{#1}{A}{Bidder Attributes}{%
		\ifstrequal{#1}{B}{Mechanisms}{%
			\ifstrequal{#1}{C}{Price of Anarchy Analysis Symbols}{%
			\ifstrequal{#1}{D}{Bayes-Nash Equlibrium Symbols}{%
					\ifstrequal{#1}{E}{Other Symbols}{%
				}}}}}%
	]}

\newcommand{\opt}{\text{opt}}
\newcommand{\greedy}{\text{greedy}}
\newcommand{\gsp}{\text{gsp}}
\newcommand{\vcg}{\text{vcg}}

\newcommand{\A}{\mathcal{A}}
\renewcommand{\P}{\mathcal{P}}
\newcommand{\D}{\mathcal{D}}
\newtheorem{clm}{Claim}
\newtheorem{cor}{Corollary}
\newtheorem{proposition}{Proposition}

\newtheorem{definition}{Definition}
\newtheorem{lemma}{Lemma}

\newtheorem{theorem}{Theorem}

\DeclareMathOperator{\argmax}{argmax}
\DeclareMathOperator{\argmin}{argmin}
\DeclareMathOperator{\E}{\mathbb{E}}

\newcommand{\cmark}{\ding{51}}%
\newcommand{\xmark}{\ding{55}}%
\newcommand{\greencheck}{\color{green}\cmark}%
\newcommand{\redx}{\color{red}\xmark}%
\renewcommand{\nompreamble}{The next list describes several symbols that are used within the body of the document.}

\newenvironment{proofsketch}{%
	\proof}{\endproof}

\newcommand\shortversion[1]{}
\newcommand\longversion[1]{#1}
\newcommand{\coloneq}{\coloneqq}

\title{Equilibria in Auctions with Ad Types}
\author{Hadi Elzayn  \and Riccardo Colini-Baldeschi \and Brian Lan \and Okke Schrijvers}

\begin{document}
	\maketitle
	\begin{abstract}
		This paper studies equilibrium quality of semi-separable position auctions (known as the \emph{Ad Types} setting \cite{adtypes}) with \emph{greedy} or \emph{optimal} allocation combined with \emph{generalized second-price} (GSP) or \emph{Vickrey-Clarke-Groves} (VCG) pricing. We make three contributions: first, we give upper and lower bounds on the Price of Anarchy (PoA) for auctions which use greedy allocation with GSP pricing, greedy allocations with VCG pricing, and optimal allocation with GSP pricing. Second, we give Bayes-Nash equilibrium characterizations for two-player, two-slot instances (for all auction formats) and show that there exists both a revenue hierarchy and revenue equivalence across some formats. Finally, we use no-regret learning algorithms and bidding data from a large online advertising platform and no-regret learning algorithms to evaluate the performance of the mechanisms under semi-realistic conditions. For welfare, we find that the optimal-to-realized welfare ratio (an empirical PoA analogue) is broadly better than our upper bounds on PoA; For revenue, we find that the hierarchy in practice may sometimes agree with simple theory, but generally appears sensitive to the underlying distribution of bidder valuations. 	
	\end{abstract}

\section{Introduction}
\input{content/intro-ec}

\section{Model and Preliminaries}
\input{content/model-ec}

\section{Price of Anarchy}\label{s:poa}
\input{content/poa-ec}

\section{Equilibrium Characterization}\label{s:eq}
\input{content/eq-ec}
\section{Empirical Study}\label{s:exp}
\input{content/exp-ec}

\section{Discussion and Open Questions}\label{s:conc}
\input{content/conc-ec}

\bibliographystyle{acm}

\bibliography{bib.bib}
\clearpage
\appendix
\input{content/app-ec.tex}
\end{document}

%% file: content/intro-ec.tex

This paper characterizes equilibrium welfare and revenue properties of various auction formats in the \emph{Ad Types setting}. The Ad Types setting \cite{adtypes} is a generalization of the standard position auction \cite{edelman2007internet, varian2007position}, which has been a workhorse in online advertising for years. In the standard position auction setting, there are multiple positions where the auctioneer can place ads. Advertisers care about receiving clicks on their ads, and the classical model posits a \emph{separable} click-through-rate (CTR) model, where ad slots have an associated discount $1\geq \delta^1 \geq \delta^2 \geq .. \geq 0$ that represents the advertiser-agnostic CTR of the slot.

The Ad Types setting \cite{adtypes} is a \emph{semi-separable} generalization of position auctions where each ad has a publicly known type\footnote{Type in the economics literature often refers to private information. That is not the case here: ad type refers to the conversion event that the advertiser cares about.}---such as `video ad', `link-click ad' or `impression ad'---and each ad type $\tau$ has its own associated position discount curve $1\geq \delta^1_\tau \geq \delta^2_\tau \geq .. \geq 0$. All ads from the same type share the same discount curve; as such, the model generalizes the position auction while maintaining more structure than a general max-weight bipartite matching problem.

Colini-Baldeschi et al. \cite{adtypes} show that in the Ad Types setting, one can compute the optimal allocation (with respect to reported bids) and associated Vickrey-Clarke-Groves (VCG) prices using an adapted version of the Kuhn-Munkres algorithm in $O(n^2(k + \log n))$ (where $n$ is the number of slots, and $k$ the number of ad types). However, there are two practical considerations that need to be taken into account: First, despite the auction-theoretical benefits of VCG, in practice online advertising platforms often use a Generalized Second-Price (GSP) payment rule \cite{ausubel2006lonely}, so it is desirable to understand the impact of using GSP pricing instead of VCG. Second, in content feeds there is often a large number of ads that are allocated, making the $O(n^2(k + \log n))$ running time prohibitive, necessitating simpler non-optimal allocation algorithms.

In this paper we investigate what happens in the Ad Types setting when we perform the allocation using either the \emph{greedy} or \emph{optimal} algorithm, and run pricing using either \emph{GSP} or \emph{VCG} semantics. In three of the four possible combinations the resulting auction is not incentive compatible, so we investigate the revenue and welfare in equilibrium.

\subsection{Contributions}
This paper makes three main contributions:
\begin{itemize}
	\item {\bf Price of Anarchy Bounds.} In Section~\ref{s:poa}, we provide Price of Anarchy upper and lower bounds in the Ad Types setting for all combinations of greedy or optimal allocation paired with GSP and VCG pricing. In particular, greedy allocation has an upper bound for Price of Anarchy of 4, regardless of the choice of pricing; for optimal allocation and GSP pricing, we give an upper bound that depends on the bidder types and number of bidders, but not valuations. We give lower bounds on the Price of Anarchy of 2 for greedy allocation with GSP pricing, 3/2 for greedy allocation with VCG pricing, and 4/3 for optimal allocation with GSP pricing. 
	\item  {\bf Small Equilibrium Characterization.} In Section~\ref{s:eq}, we analytically characterize the existence of Bayes-Nash equilibrium in the simple case of two bidders, two slots, and valuations distributed uniformly over the unit interval.\footnote{While this may appear a very special case, explicit equilibrium characterization in auctions is notoriously complex. Most famously, in Vickrey's original paper \cite{vickrey1961counterspeculation} he posed an open problem to characterize the equilibrium of a two-player first-price auction with uniform valuations in $[a_1, b_1]$ and $[a_2, b_2]$. The problem remained unsolved until nearly 50 years later \cite{kaplan2012asymmetric}!} In equilibrium, the greedy allocation with GSP pricing produces and equivalent amount of revenue to the optimal allocation with VCG pricing, and that this revenue is larger than the revenue produced by either of the other possible mechanism (which are also equivalent to each other).
	\item {\bf Evaluation on Realistic Data.} The small-equilibrium characterizations are interesting, but in order to understand if the results are representative of larger instances, we learn equilibria for bidding data from a large online advertiser in Section~\ref{s:exp}. We draw (normalized and anonymized) advertiser bids in various settings and equip advertisers with no-regret learning algorithms; when players use such algorithms, the empirical average of play is known to converge to \emph{coarse correlated equilibria}.  We find that for the most part equilibria on real data do not behave identically to the two bidder two slot uniform valuations case, but rather show a steeper hierarchy of revenue and welfare that conforms with intuition.   
\end{itemize}

\subsection{Related Literature}
\paragraph{Position Auctions.} Position auctions have long been the workhorse in online advertising. The seminal works of Edelman et al. \cite{edelman2007internet} and Varian \cite{varian2007position} first proposed the separable model of the position auction---and described the generalized second-price (GSP) auction in this model---and showed that for GSP there exists an ex-post Nash equilibrium that is equivalent to the VCG outcome. Gomes and Sweeney \cite{gomes2014bayes} showed that GSP does not always admit a Bayes-Nash equilibrium. There is also a history of exploring alternative pricing rules for position auctions; for example Chawla and Hartline \cite{chawla2013auctions} study generalized first-price (GFP) semantics for position auction and show that for independent and identically distributed (IID) valuations the equilibrium is unique and symmetric.

\paragraph{Price of Anarchy and Smoothness.} Since explicit equilibrium computation in auction is challenging, people have focused on Price of Anarchy bounds, i.e. using the equilibrium conditions to give bounds on the welfare in \emph{any} equilibrium. Paes Leme and Tardos \cite{leme2010pure} were the first to give Price of Anarchy bounds for GSP. A common approach to proving Price of Anarchy bounds is to use the smoothness framework proposed by Roughgarden \cite{roughgarden2015intrinsic,roughgarden2017price}, though GSP is not smooth in this sense. Lucier and Paes Leme \cite{lucier2011gsp} and Caragiannis et al. \cite{caragiannis2015bounding} instead show that one can use a \emph{semi-smoothness} condition and they give almost tight Price of Anarchy bounds for GSP. Smoothness has also been applied to other payment rules, such as GFP by Syrgkanis and Tardos \cite{syrgkanis2013composable}.

\paragraph{Complex Ad Auctions.} There is a body of work that explores relaxing the separability assumption in position auctions. Our work is based on the Ad Types setting formalized by Colini-Baldeschi et al. \cite{adtypes}. When each ad is its own type, this model is identical to the one with arbitrary action rates that are still independent between ads, which has been studied before by Abrams et al. \cite{AGV07}, Carvallo and Wilkens \cite{CW14} and Wilkens et al. \cite{cavallo2018matching}. To our knowledge, no equilibrium characterizations or Price of Anarchy bounds are known in these settings. The closest is a paper by Colini-Baldeschi et al. \cite{colini2020envy} that studies the relationship between envy, regret and social welfare loss in the Ad Types setting for an alternative version of GSP called ``extended GSP'' using the same semi-smoothness framework as proposed by Caragiannis et al. \cite{caragiannis2015bounding}.

%% file: content/model-ec.tex

\paragraph{Advertisers.} There are $n$ advertisers (each associated with a single ad) competing for $m$ (ordered) slots. Each ad has a publicly known type $\tau_i$, such as `video ad', `link-click ad' or `impression ad'. Ad $i$ of type $\tau_i$ has value-per-conversion $v_i$. Ads of different types have different conversion events, e.g. for a link-click ad the conversion event is a link click and for a video ad the conversion event is the user watching a video ad.
\paragraph{Slots.} Slots are indexed by integers which increase moving down the feed. (So ``lower'' slots have higher indices.) Ads in lower slots see fewer conversions, and we consider a semi-separable model\footnote{The model is semi-separable since ads of the same type share the same discount curve, but ads of different types do not.} to capture this effect: for ads of type $\tau_i$, we can write $\Pr[\text{conversion on ad }i\text{ (of type }\tau_i\text{) in slot }s] = \delta^s_{\tau_i} \cdot \beta_i$ where $\delta^s_{\tau_i} $ is the slot effect for a particular ad type $\tau_i$ (e.g., the probability that a user will watch a video ad if it is shown in the $s$th slot) and $\beta_i$ is the advertiser effect. We assume without loss of generality that the advertiser effect has been included in the advertiser's value, i.e., if the value-per-conversion of the advertiser is $v_i'$, then $v_i = \beta_i\cdot v'_i$. Since the advertisers effectively \emph{discount} their value for the slot by $\delta_{\tau_i}^s$, we call $\left(\delta_{\tau_i}^{s}\right)$ for all slots $s$ the \emph{discount curve}. 

\paragraph{Bidding and Payoffs.} 
Advertisers submit a single bid $b_i$ for a conversion, which may or may not be their true valuation $v_i$. They are charged price $p_i$ (calculated by the auction) if a conversion happens, so in expectation they are charged $\delta_{\tau(i)}^s p_i$. Thus, the expected payoff of an advertiser for a given slot at a given price is
$u_i \left(s, p_i\right) = \delta_{\tau(i)}^s \left(v_i - p_i\right).$

\paragraph{Discount Curves.} We assume that discount curves \emph{monotonically decrease with the slot index}: that is, $1\geq \delta^1_{\tau_i} \geq \delta^2_{\tau_i} \geq ... \geq 0$. We will say that $s\succeq_i s'$ (read as $i$ prefers $s$ to $s'$) if $\delta_{\tau_i}^s \geq \delta_{\tau_i}^{s'}$. Since the conversion probability decreases moving down the feed for all types, advertisers agree on their preference between any pair of slots, so we can drop the subscript and simply use $\succeq$. Notice that since slots lower down the feed are indexed by higher numbers, $s \succeq s' \iff s \leq s'$; we will often speak in terms of preference in order to avoid confusion. In some restricted settings, we consider \emph{geometric} discount curves that can be written as $\delta_{\tau}^s = c\cdot \delta^s$ for some fixed $c,\delta$, where $s$ is an exponent on the right hand side. 

\paragraph{Auction Algorithms.}
Any auction must answer two questions: who gets what (allocation), and much how do they pay (pricing).  We  use $\A:\mathbf{b} \to \mathbf{s}$ to designate allocation algorithms, and $\P: \A, \mathbf{b} \to \mathbf{p}$ to designate pricing algorithms. Here, $\mathbf{b}$ is a vector of bids and $\mathbf{s}$ is a vector of slot assignments. In other words, an allocation algorithm $\A$ maps bid vectors to slot vectors. A pricing algorithm $\P$, however, takes both a vector of bids and an allocation \emph{algorithm} $\A$. 
Thus the pricing algorithm is a meta-algorithm, rather than a particular algorithm.

We refer to a pair $(\A, \P_\A)$ as an \emph{auction mechanism}.  In this paper we consider all combinations of two allocation algorithms and two pricing meta-algorithms:

\begin{itemize}
	\item {\bf Greedy (Allocation)} The \emph{greedy} allocation begins with the highest slot, and among non-allocated bidders allocates the bidder whose \emph{discounted} bid is highest (that is, $\argmax_{i \in \mathcal{U}_s} \delta_{\tau(i)}^{s}b_i$, where $\mathcal{U}_s$ is the set of unallocated bidders as of the time slot $s$ is reached). For the Ad Types setting, the greedy algorithm generally does not yield the optimal allocation (see e.g. Example 1.1 in \cite{adtypes}).
	
	\item {\bf Optimal (Allocation)} The \emph{optimal} allocation computes the max-weight bipartite matching between ads and slot (where edge weights are discounted bids $\delta_{\tau(i)}^{s}b_i$), e.g. using the Kuhn-Munkres algorithm \cite{kuhn1955hungarian, munkres1957algorithms}. 
	\item {\bf GSP (Pricing)} The \emph{Generalized Second Price} pricing rule executes the principle that a bidder pays the minimum bid under which they retain the slot they were assigned to, i.e. for allocation algorithm $\A$ and bids $\mathbf{b}$: $\left[\P_\A(\mathbf{b})\right]_i := {\argmin}_{b: \A(b, \mathbf{b}_{-i})_i = \A(\mathbf{b})_i} b$. Computing this bid is straightforward for the greedy allocation algorithm, while for the optimal algorithm we use the method of Carvallo et al \cite{cavallo2018matching}.
	\item {\bf VCG (Pricing)} The \emph{Vickrey-Clarke-Groves} pricing rule \cite{vickrey1961counterspeculation,C71,G73} executes the principle that a bidder should pay their \emph{externality}, i.e. for an allocation algorithm $\A$ and bids $\mathbf{b}$:
	$
	\left[\P(\mathbf{b})\right]_i = \sum_{j \neq i} \delta_{\tau(j)}^{\A(\mathbf{b}_{-i})_j} b_j - \sum_{j \neq i} \delta_{\tau(j)}^{\A(\mathbf{b})_j } b_j.
	$
	When $\A$ is the optimal allocation algorithm this yields the standard VCG algorithm. When $\A$ is the greedy allocation algorithm, the resulting mechanism is not incentive compatible.
\end{itemize}

Given  an auction $(\A, \P_\A)$, bids $\bf b$, and valuations $\bf v$, the \emph{social welfare} is $\text{W}(\A, \P_\A, {\bf b}, {\bf v}) = \sum_i \delta^{\A(\bf b)_i}_{\tau(i)} \cdot v_i$ and the \emph{revenue} is $\text{Rev}(\A, \P_\A, {\bf b}, {\bf v}) = \sum_i \delta^{\A({\bf b})_i}_{\tau(i)} \cdot \P_\A({\bf b})_i$. At the risk of restating the obvious, notice that the auctioneer can only observe reported bids, not true valuations; hence, to the extent that each mechanism computes an ``optimal'' allocation, it is optimal with respect to the \emph{bids}, not values.  For non-incentive compatible mechanisms, these will not in general coincide, and we must take care in the analysis not to conflate the two; we will emphasize ``apparent'' with the ``hat'' symbol, e.g. we denote the apparent social welfare with respect to the bids as $\widehat{W}$. 

\paragraph{Additional Notation.} To indicate vectors, we will use bold font: e.g. we denote the vector of bids as $\mathbf{b}$. We will use subscripts to denote a particular component, e.g. $\mathbf{b}_i$ is the $i$th component of $\mathbf{b}$. At the risk of overloading notation, we will also use $i$ as a subscript to track scalar functions for particular player. So, for instance, we can write $b_i$ for player $i$'s bid, or $\mathbf{b}_i$, depending on whether we are arguing about the auctioneer's or player's perspective. (The  meaning of the subscript should be clear from context.) Also, we use the standard $-i$ subscript to indicate ``all but the $i$th component'' of a vector. We will also use an analogous $-i$ superscript for scalar functions, e.g. we write $W^{-i}$ for the scalar welfare of all players but $i$. \footnote{To streamline notation, we will omit extra parentheses when writing the modified component and the ``all-but-ith'' component together where a vector would be required. For instance, $\A(b_i',\mathbf{b}_{-i})$ instead of $\A((b_i',\mathbf{b}_{-i}))$.}. 

Since we consider multiple allocation and pricing formats, we write $\pi_{\A}(s, \mathbf{b})$ to indicate the \emph{player} in \emph{slot} $s$ when $\mathbf{b}$ is the bid profile and $\A$ is the allocation algorithm. We will suppress the ${\A}$ when it is clear from context. 
We use $\sigma_{\A}(i, \mathbf{b})$ to indicate the \emph{slot} that player $i$ receives when the bid profile is $\mathbf{b}$ and the allocation algorithm is $\A$. We will sometimes overload notation to write $\tau(i)$ as function returning player $i$'s ad type $\tau_i$; this will be useful when referring not to a specific player but rather to an arbitrary occupant of a given slot. We can also compose some or all of these together. For example, $\tau(\pi_{\A}(\sigma_{\A}(i,\mathbf{b}),\mathbf{b}')$ is the type of the player assigned to the slot that $i$ receives under the allocation algorithm $\A$ and bid vector $\mathbf{b}$ given that the bid vector is \emph{instead} changed to $\mathbf{b}'$.  Finally, we will denote the \emph{optimal} allocation vector given a bid profile $\mathbf{v}$ with the Greek letter $\boldsymbol{\nu}$. That is:
\begin{align*}
	\boldsymbol{\nu} :=\underset{\boldsymbol{\sigma} \in S_{n}}{ \argmax} \sum_{i =1}^{n} \delta_{\tau(i)}^{\boldsymbol{\sigma}_i} v_i,
\end{align*} 
where $S_n$ is the set of all permutations of bidders. Maintaining our font conventions, we use $\nu(i)$ for the mapping of $i$ to the slot index he is assigned under $\boldsymbol{\nu}$.

\subsection{Solution Concepts and Learning} 
Each mechanism induces a game between agents that act strategically, so the \emph{equilibrium concept} is an important modeling choice. In this paper we present equilibrium results for both full-information and Bayes-Nash equilibria:
\begin{definition}[Nash Equilibrium]\label{def:eq-nash}
A bid profile $\mathbf{b}$ is pure strategy \emph{Nash equilibrium} if for each player $i$:
$
u_i(\mathbf{b}) \geq u_i(b', \mathbf{b}_{-i})
$ 
for all pure strategies $b'$.	
\end{definition}

\begin{definition}[Bayes-Nash Equilibrium]\label{def:eq-bayes}
	For a known value distribution $\mathcal{V}$, the vector of mappings $\mathbf{b}(\mathbf{v})$ is a Bayes-Nash equilibrium if for every player $i$:
\begin{align*}
	\E_{\mathbf{v}\sim \mathcal{V}}[u_i(\mathbf{b}(\mathbf{v}))] \geq \E_{\mathbf{v}\sim \mathcal{V}}[u_i(b_i'(v_i), \mathbf{b}_{-i}(\mathbf{v}_{-i}))]
\end{align*}
	for any other mapping ${b'_i}(v_i)$.
\end{definition}
For each of these equilibrium notions, an $\epsilon$-\emph{approximate} version is obtained by allowing the definitional inequality to be violated by no more than $\epsilon$. A bid profile where no bidder can improve their payoff by more than $\epsilon$ is an $\epsilon$-approximate Nash equilibrium.

A Bayes-Nash equilibrium is $\emph{linear}$ if $b_i(v_i) = \beta_i v_i$ for some $\beta\geq 0$. We say a bidder is \emph{conservative} if he does not bid above his value, and an equilibrium is conservative if it does not prescribe bidding above one's value. For some results in Section \ref{s:poa}, we will assume that bidders are conservative.

In general, it may be difficult or impossible to analytically characterize equilibria in more complicated settings. Thus, in  Section \ref{s:exp}, we turn to \emph{learning} equilibria using no-regret learning algorithms on data drawn from realistic valuation distributions. This approach, while powerful, is not guaranteed to recover either Nash or Bayes-Nash equilibria, but instead the more general notions of \textit{Coarse Correlated Equilibrium} (CCE) and \textit{Bayesian Coarse Correlated Equilibrium} (BCCE). As we do not rely on these notions for our analytical results, we defer the definitions of these concept to Section \ref{sec:appexperiments}. 

For each equilibrium concept, there may be multiple equilibria with different welfare. The Price of Anarchy (PoA) captures the worst-case\footnote{We also use the term ``empirical  PoA'' to describe the ratio of average realized welfare to optimal welfare when speaking about specific or empirical cases. Strictly speaking the PoA is only the worst-case value, but the meaning should be clear.} welfare compared to the optimal welfare knowing the valuations. Here, we write its definition adapted to our setting:
\begin{definition}[Price of Anarchy] \label{def:poa}
	The Price of Anarchy is
	 $$  \text{PoA(Nash)} \coloneqq \max_{{\bf b} \in E} \frac{\sum_i \delta^{\nu(i)}_{\tau(i)}\cdot \mathbf{v}_i}{\E[\sum_i \delta^{\A({\bf b})_i}_{\tau(i)} \cdot \mathbf{v}_i]} 
	 $$ 
	 where $E$ is the set of Nash equilibria for $(\A, \P_\A)$ and the randomness is over the strategy distributions. A similar definition can be made for a Bayesian PoA with randomness over the valuations. 
\end{definition}

%% file: content/poa-ec.tex
In this section, we provide characterizations of upper and lower bounds on the Price of Anarchy for (Greedy, GSP), (Greedy, VCG), and (Opt,GSP) with conservative bidders\footnote{We omit (Opt, VCG); that bidding truthfully is a dominant strategy suggests alternative equilibria are unlikely in practice.}. For upper bounds on the Price of Anarchy, we leverage the \emph{semi-smoothness} framework of \cite{caragiannis2015bounding}, itself a generalization of the \emph{smoothness} framework of \cite{roughgarden2015intrinsic}.  For lower bounds, we construct examples of equilibria\footnote{Conservative bidders likely match reality when agents are not sophisticated or uncertain. This restriction is common in the literature, e.g. \cite{caragiannis2015bounding}, and can be interpreted as strengthening our PoA lower bounds and weakening our upper bounds.} that achieve less welfare than the optimal. For results that are primarily ancillary or require involved proofs, we provide proof sketches, and defer full proofs to an expanded online version of the paper. 

For (Greedy, GSP) and (Greedy, VCG), we give a universal result - that is, under no requirements besides being in the Ad Types setting, and this result matches known upper and lower bounds for the position auction (though our bounds are not yet as tight). 
For (Opt, GSP), we 
provide \emph{instance-optimal} bounds; here, instance-optimal means allowing for dependence on the discount curves and number of slots but  \emph{not} over bidder valuations. It is very likely that our upper bounds on the Price of Anarchy in this setting are too pessimistic; we leave improvement of these bounds to future work.

\begin{table}[h]

	\begin{subtable}[t]{0.45\textwidth}
		\centering
	\begin{tabular}{ccc}
		& GSP & VCG\\
		\toprule
		Greedy & $2$& 3/2 \\
		Opt & 4/3& NA
	\end{tabular}
	\caption{Lower bounds on PoA. \label{tab:poalower}\\ }
\end{subtable}%
\hfill
\begin{subtable}[t]{0.45\textwidth}
	\centering
	\begin{tabular}{ccc}
		& GSP & VCG\\
		\toprule
		Greedy & $4$& $4$ \\
		Opt & $2 + 2(n-1){\frac{\delta^{\max}}{\delta_{\min}}}^{*}$& NA
	\end{tabular}
	\caption{Upper bounds on PoA. $^{*}$ denotes instance-optimal bounds.\label{tab:poauper}}
	\end{subtable}
\caption{Price of Anarchy Bounds}
\end{table}

Our technique in each case will be to show that the game induced by the auction format and any valuation profile is \emph{semismooth}, in the following sense:
\begin{definition}[Semismooth game \cite{caragiannis2015bounding}]\label{def:smooth}
	We say that a game is $(\lambda,\mu)$-semismooth if there exists a (possibly randomized) strategy $b'$ which depends only on a player's valuation such that:
	\begin{align*}
	u_i(b'(v_i), \mathbf{b}_{-i}) \geq \lambda \sum_{i} \delta_{\tau_i}^{\nu(i)}\mathbf{v}_i - \mu \sum_{i} \delta_{\tau_i}^{\sigma(i,\mathbf{b})} \mathbf{v}_i. 
	\end{align*}
	for all bid profiles $\mathbf{b}$ and all valuation vectors $\mathbf{v}$.
	\end{definition}
A game can be shown to be semismooth by showing that the
the following inequality holds:
\begin{align*}
u_i(b_i', \mathbf{b}_{-i}) \geq \lambda \delta_{\tau(i)}^{\nu(i)} v_i - \mu \delta_{\tau(\pi(\nu(i),\mathbf{b}))}^{\nu(i)} v_{\pi(\nu(i),\mathbf{b})},
\end{align*}
since if it holds, summing over players gives exactly the defining condition of semismoothness. And semismoothness directly yields Price of Anarchy bounds using the following theorem, from \cite{caragiannis2015bounding}:
\begin{theorem}[\cite{caragiannis2015bounding}]
	Suppose a game is $(\lambda, \mu)$-semismooth, and social welfare is at least the sum of player utilities. Then its Price of Anarchy is upper bounded by $\frac{\mu+1}{\lambda}$. 
	\end{theorem}

\subsection{Greedy Allocation Proof Recipe}
A common proof structure applies to both (Greedy, GSP) and (Greedy, VCG), because of their shared allocation algorithm and the fact that both pricing algorithms, when coupled with greedy allocation, guarantee that bidders are never overcharged. It is similar to the proof found in \cite{caragiannis2015bounding}, but with additional subtlety due to the differing discount factors. 

To handle this subtlety, we will use the following Lemma:
\begin{lemma}[Partial Monotonicity] \label{lem:partialmon}
	Suppose that $\bf b, \bf b'$ are two bid profiles that only differ in element $i$, and $b_i'>b_i$. Let $\sigma$,$\sigma'$ be the slots which $i$ was assigned under $\bf b$, $\bf b'$ respectively. Then under greedy allocation, we have that for each slot $s$ strictly above $\sigma$:
	\begin{align*}
	\delta_{\tau(\pi(s, \mathbf{b'}))}^{s} \mathbf{b}_{\pi(s,\mathbf{ b'})} \geq \delta_{\tau(\pi(s,\mathbf{b}))}^{s} \mathbf{b}_{\pi(s, \mathbf{b})}^{s}
	\end{align*}
\end{lemma}
Informally, this lemma merely states that if bidder $i$ deviates upwards from his bid under $\mathbf{b}$, the value obtained by players in the slots above his placement under $\mathbf{b}$ can only increase. To see why this is true, recall that the greedy algorithm allocates from top to bottom. So for every slot $s$ between $\sigma'$ and $\sigma$  (not including $\sigma$), the bidders considered when $s$ was assigned under $\mathbf{b}$ remain unallocated when considering $s$ under $\mathbf{b}'$; hence $\pi(s, \mathbf{b'})$ (i.e. whoever is assigned to $s$ under $\mathbf{b}$ must have at least as high of an effective value as $\pi(s, \mathbf{b}')$. See Appendix \ref{s:app-smoothproofs} for a more formal proof.

Now we are ready to state and prove our theorem.
\begin{theorem}[Semi-Smoothness for Greedy Algorithms]\label{thm:receip}
		Let $(\A, \P_\A)$ be an auction mechanism. Suppose that
		\begin{enumerate}
			\item $\A$ is the greedy algorithm, and
			\item For any bid profile $\mathbf{b}$, for every bidder we have:
			\begin{align*}
			\P_{\A}(\mathbf{b})_i \leq \mathbf{b}_i
			\end{align*}
		\end{enumerate}
	Then $(\A,\P_{\A})$ is (1/2,1)-Semi Smooth.
	\end{theorem}

\begin{proof}
		Recall that if we can show that for any bid profile:
	\begin{align*}
	u_i(b_i', \mathbf{b}_{-i}) \geq \delta_{\tau(i)}^{\nu(i)} \frac{v_i}{2} - \delta_{\tau(\pi(\nu(i),\mathbf{b}))}^{\nu(i)} v_{\pi(\nu(i),\mathbf{b})}
	\end{align*}
	then we will be done. 
	So suppose $\mathbf{b}$ is a bid profile, and consider a deviation strategy of bidding half one's value. (Notice first off that such a deviation guarantees a deviating bidder non-negative utility by Property 2.) Fix bidder $i$. Under this unilateral deviation, $i$ receives $\sigma' := \A(b_i',\mathbf{b}_{-i})_i$. There are two casese to consider: either $\sigma' \succeq \nu(i)$ (i.e. $\sigma'$ is $\nu(i)$ or better) or $\sigma' \prec \nu(i)$ ($\sigma'$ is strictly worse than $\nu(i)$.)  If the first case holds, we achieve the desired inequality since:
		\begin{align*}
	u_i(b_i', \mathbf{b}_{-i}) &= \delta_{\tau(i)}^{\sigma'} v_i - \delta_{\tau(i)}^{\sigma'} \mathcal{P}_i(b_{i}',\mathbf{b}_{-i'}, \A) 
	\\
	&\geq \delta_{\tau(i)}^{\sigma'} v_i -\delta_{\tau(i)}^{\sigma'} \frac{v_i}{2} = \delta_{\tau(i)}^{\sigma'} \frac{v_i}{2}\\
	&\geq \delta_{\tau(i)}^{\nu(i)}\frac{v_i}{2}\\
	&\geq \delta_{\tau(i)}^{\nu(i)} \frac{v_i}{2} - \delta_{\tau(\pi(\nu(i), \mathbf{b}))}^{\nu(i)} v_{\pi(\nu(i), \mathbf{b})},
	\end{align*}
	where the first inequality follows by no-overcharging and the others by assumption or trivially.
	
	Now suppose that instead, $\A(b_i', \mathbf{b}_{-i})_i = \sigma' \prec \nu(i)$.  We split this into two subcases. 
	In the first subcase, $\frac{v_i}{2} \geq b_i$, i.e. $b_i'$ is an upward deviation that results in $i$ receiving $\sigma'$ below $\nu(i)$. 
	Combining $(b_i',\mathbf{b}_{-i})$ into $\mathbf{b}'$, we can write:
	\begin{align*}
		\delta_{\tau(\pi(\nu(i), \mathbf{b}'))}^{\nu(i)} \mathbf{b}_{\pi(\nu(i),\mathbf{b}')}'  \geq \delta_{\tau(i)}^{\nu(i)} \frac{v_i}{2}.
	\end{align*}
	This follows because $i$ was unallocated when $\nu(i)$ was considered, so if this did not hold, the greedy allocation would have allocated $i$ to $\nu(i)$ instead of $\pi(\nu(i),\mathbf{b}')$.
	
	Now, notice that we can view $b_i$ as a downward deviation from $b_i'$, and a downward deviation cannot affect the allocation choices of any of the slots above its place before the deviation, including $\nu(i)$. But that means that the allocated bidder to $\nu(i)$ is the same under $\mathbf{b}$, so the inequality above also implies that:
	\begin{align*}
			\delta_{\tau(\pi(\nu(i),\mathbf{b}))}^{\nu(i)} \mathbf{b}_{\pi(\nu(i),\mathbf{b})} \geq \delta_{\tau(i)}^{\nu(i)} \frac{v_i}{2}.
	\end{align*}
		Then using no-overcharging, we again have that
		\begin{align*}
		\delta_{\tau(i)}^{\nu(i)} \frac{v_i}{2} - \delta_{\tau(\pi(\nu(i),\mathbf{b}))}^{\nu(i)} \mathbf{b}_{\pi(\nu(i),\mathbf{b})}\leq 0 \leq u_i(b_i', \mathbf{b}_{-i})
		\end{align*}
		
	Finally, suppose $\frac{v_i}{2} < b_i$. As before, we must have that $\pi(\nu(i) \mathbf{b}')$ must have at least as high effective value as $i$. To see that $\pi(\nu(i), \mathbf{b})$ \emph{also} has at least as high effective valuation as $i$, 
	 notice that we can view $b_i$ as an upward deviation from $b_i'$. By assumption, $\sigma'\prec \nu(i)$, so Lemma \ref{lem:partialmon} implies that in moving to $\mathbf{b}$, the values of bidders in slots above $\sigma'$, which include $\nu(i)$, must increase. But then we have again that:
	\begin{align*}
				\delta_{\tau(\pi(\nu(i),\mathbf{b}))}^{\nu(i)} \mathbf{b}_{\pi(\nu(i),\mathbf{b})} \geq \delta_{\tau(i)}^{\nu(i)} \frac{v_i}{2}.
	\end{align*} 
	and the desired inequality follows as before.
\end{proof}

\subsection{Greedy Allocation and GSP Pricing}

\begin{theorem} 
Let $(\A, \P_{\A})=$ (Greedy, GSP). Then the Price of Anarchy is at most 4. 
\end{theorem}
\begin{proof}
	First, by assumption, $\A$ is Greedy. Second,  generalized second price will not charge a bidder more than their bid since under the greedy algorithm, the winner of a slot has a higher effective bid than the second bidder's bid, which is what they are charged. Hence, the conditions of Theorem \ref{thm:receip} are satisfied, so the induced game is $(\frac{1}{2}, 1)$-semismooth and the bound follows.
\end{proof}

On the other hand, we can show that the Price of Anarchy is \emph{at least} 2. 
\begin{theorem}
	Let $(\A, \P_{\A})=$(Greedy, GSP). Then the Price of Anarchy is at least 2. 
\end{theorem}
\begin{proof}
	Consider the following example: there are 2 slots and 2 bidders, one of type A and one of Type B. Let $\boldsymbol{\delta_A}=(1,0)$, $\boldsymbol{\delta_B}=(1,1)$, and let $v_A = (1-\epsilon)v_B$, $\epsilon>0$. Then the allocation $(A,B)$ gets payoff $v_A+v_B=(2-\epsilon) v_B$, while the allocation $(B,A)$ gets welfare $v_B$. 
	
	We claim that the following is an equilibrium: A bids $0$ and B bids $v_B$, giving the allocation $(B,A)$. To see that this is an equilibrium, notice that if these are the bids, $b_B>b_A$, so B will be given the first slot at a price of $b_A=0$ for a total payoff of $v_B$. Since price is bounded below by $0$, B could not gain by deviating any lower. On the other hand, in the second slot, $A$ gets no value, but also is not charged, for a payoff of $0$. To change anything, A would have to change the allocation, and so bid above $b_B = v_A$ - but then she would get a payoff of $v_A-v_B = (1-\epsilon)v_B- v_B \leq 0$; hence she also would not like to switch. And note that since $0 \leq b_A$ and $v_B \leq v_A$, neither bidder is overbidding. But thus we see that 
	\begin{align*}
		\frac{OPT}{EQ} = \frac{v_A+v_B}{v_B} = \frac{(2-\epsilon)v_B}{v_B} = 2-\epsilon
		\end{align*}
	
	and so the Price of Anarchy can be made arbitrarily close to 2. 
	
\end{proof}
Note the equilibrium described is not unique - for instance, $b_A= (1-\epsilon)v_B$, $b_B=v_B$ would also be an equilibrium that achieves the same allocation. 

\longversion{
For some intuition as why such a simple example can get a bad price of anarchy, notice that two slot case can be mapped to a standard second price auction for the first slot, where one bidder has a good outside option and the other doesn't. By including the outside options, a socially-minded auctioneer could do significantly better than just considering the bid and valuations of the item in question. 
}
\subsection{Greedy Allocation and VCG Pricing}
In this section, we consider the Price of Anarchy when $(\A, \P_{\A})$ is (Greedy, VCG).  Again, using greedy allocation guarantees the first condition of Theorem \ref{thm:receip}. It is not obvious that bidders will not be overcharged. It is, however, true, as we show in the following Lemma:
\begin{lemma}\label{lem:greedyvcgnooc}
Let $(\A, \P_{\A})$ be the greedy algorithm with VCG pricing. We claim that for every bidder, their charge will not exceed their effective bid.
\end{lemma}
\longversion{
\begin{proof}
	We will prove this by strong induction. First, we relabel the bidders so that Bidder $i$ is in Slot $i$ post-allocation. Now, consider the removal of bidder $i$. First notice that this will not affect the assignment to any $i'$ \emph{above} $i$. So any price that $i$ must pay will come from the externalities he imposes on $i'>i$. 

	Now, we claim that the following is true:
	\begin{align}\label{eq:payment}
	p_i = p_{j^*} + \left(\delta_{\tau(j^*)}^{i} - \delta_{\tau(j^*)}^{j^*}\right) b_{j^*}
	\end{align}
	where $j^*$ is the bidder that is assigned to Slot $i$ in the absence of Bidder $i$. (In keeping with our formal notation, $j^*\coloneqq \pi(i,(\mathbf{b}_{-i}))$.) 
	
	To see that this is true, imagine re-running the auction without $i$ included. Slots $1...i-1$ will be allocated the same way, and then at Slot $i$ some bidder $j^*$ will be allocated that would have been allocated further down had $i$ been included. Now, as $j^*$ moves up to $i$, he has not affected the winning bid calculations of all slots \emph{between} $i$ and $j^*$ relative to what they were when $i$ was included.
	
	 But that means that the only externalities that $i$ imposes are those on $j^*$ and below. Note that when we consider $j^*$ taking the slot of $i$, the arrangement of the bidders below $j^*$ will be exactly the same as if $j^*$ were the initially removed bidder instead of $i$ - but this is exactly the re-arrangement that generates the price $j^*$ pays. Hence, $i$'s total payment is the payment of $j^*$ plus the externality he imposes on $j^*$, which is $b_{j^*}(\delta_{\tau(j^*)}^{i} - \delta_{\tau(j^*)}^{j^*})$. But this is exactly what is claimed in Equality \ref{eq:payment}. Then we can write:
	
	\begin{align*}
	p_i = &p_{j^*} + b_{j^*} (\delta_{j^*}^i-\delta_{j^*}^{j^*})\\&=
	p_{j^*} - \delta_{j^*}^{j*} b_{j^*} + \delta_{j^*}^i b_{j^*} 
	\end{align*}
	
Now we invoke strong induction. Suppose that all bidders below $i$ are not overcharged, i.e. $\forall j$ assigned to a slot below $i$'s, $p_j \leq \delta_{j}^{j} b_j$. Then in particular, $p_j^* - \delta_{j^*}^{j^*} b_{j^*} \leq 0$, so that we conclude:
	\begin{align*}
	p_i = p_{j^*} - \delta_{j^*}^{j*} b_{j^*} + \delta_{j^*}^i b_{j^*} \leq \delta_{j^*}^ib_{j^*} \leq \delta_i^i b_i
	\end{align*}
	where the last inequality follows by the fact that $i$ was chosen over $j$ for Slot $i$. Finally, note that Bidder $n$ pays 0, since there are no bidders below him to exert an externality on; thus, applying strong induction starting from the bottom yields the claim.
	
\end{proof}
}
\shortversion{
	\begin{proofsketch}
To prove this claim, we characterize the price paid by every bidder. In particular, it turns out to have a simple form: the price $i$ pays is exactly sum of the externality exerted (measured in terms of effective bid, not value, since again that is all the algorithm has access to) on the bidder, say $j$, who would have taken $i$'s spot in his absence \emph{plus} the price that $j$ is charged. This turns out to be equivalent to $j$'s value for $i$'s slot minus $j$'s surplus for being in his own slot. Then, we can use strong induction: we suppose that all bidders below $i$ (including j) are not overcharged (which is trivially true for the last bidder), which means that their surplus is positive, and hence, the price $i$ pays must be less than $j$'s value for the slot. And since $i$ won the slot, $i$ has higher value for the slot than $j$, so $i$ is not being overcharged. 
\end{proofsketch}
}
Lemma \ref{lem:greedyvcgnooc} allows us to conclude that (Greedy, VCG) satisfies the conditions of Theorem \ref{thm:receip}, yielding the following Theorem:
\begin{theorem} \label{thm:ubgreedyvcg}
	Let $(\A, \P_{\A})=$ (Greedy, VCG). Then the Price of Anarchy is at most 2. 
\end{theorem}
For lower bounds, we again find a suboptimal equilibrium:
\begin{theorem} \label{thm:lbgreedyvcg} Let $(\A, \P_{\A}) =$ (Greedy, VCG). The Price of Anarchy is at least $3/2$. 
	\end{theorem}

\begin{proofsketch}
Let $v_A=1+\epsilon$, $v_B = 1$, $v_C = 1-\epsilon$. Let $\mathbf{\delta_A} = (1, 1,1-2\epsilon)$, $\mathbf{\delta_B}=(1,1,0)$, $\mathbf{\delta_C}=1, \epsilon, \epsilon^2$. 

The welfare of $(C,B,A)$ is $3-2\epsilon -2\epsilon^2$, while the welfare of $(A,B,C)$ is $2+\epsilon+\epsilon^2 -\epsilon^3$. Suppose that each player bids their value, ie:
\begin{align*}
\mathbf{b}^* = (b_A,b_B,b_C) = (v_A,v_B,v_C) = (1+\epsilon, 1, 1-\epsilon).
\end{align*}
Then the allocation $(A,B,C)$ will be chosen, despite being suboptimal. The rest of the proof consists in showing that no player has an incentive to unilaterally deviate. Note that we need only consider deviations that change the selected ordering; the values and bids are chosen in such a way that any bidder that could improve their ordering would suffer too high a high price, and any bidder that could lower their ordering prefers where they are at the price they are paying. Together, this means that we have exhibited an equilibrium where $OPT/EQ$ can be made arbitrarily close to $3/2$ by taking $\epsilon$ small. 
\end{proofsketch}

\subsection{Optimal Allocation and GSP Pricing}
In the case of Optimal Allocation and GSP pricing, we will obtain a smoothness result that depends on the largest and smallest discounts and the number of bidders, but not on the valuation profile. The result is as follows:
\begin{theorem}\label{thm:smooth-optgsp}
	Suppose $(\A, \P_{\A})$ is optimal allocation and GSP pricing. Then the game between bidders is	$(\frac{1}{2}, \frac{\delta^{\max}}{\delta^{\min}}(n-1))$-semismooth.
\end{theorem}
To prove this result, we begin by observing that GSP pricing will never charge a bidder more than his effective bid\shortversion{, which is easy to see}. Formally: 
\begin{lemma}\label{lem:opt-gsp-upperboundprice} 
	In (Opt, GSP), bid upper bounds price. \end{lemma}
\longversion{
\begin{proof}
	By definition, the GSP price is the minimum the bidder could have bid and still earned the slot given the allocation algorithm and the other bids. But in particular, they could have bid exactly their bid and received their slot (because they did). Hence, the minimum they could have bid to receive the slot can never be more than whatever they actually bid. 
\end{proof}
}

\begin{proof}[Proof of Theorem \ref{thm:smooth-optgsp}]
	Again, assume that the deviation is to $b_i' = v_i/2$, and show that:
	\begin{align*}
	\sum_{i} u_i(b_i',\mathbf{b}_{-i}) \geq \frac{1}{2} \text{OPT} - \text{SW} (\mathbf{b})
	\end{align*}
	Now, let $\nu(i)$ be the slot of $i$ under the optimal allocation. Then if $i$ receives some slot $\sigma(i, b_i', \mathbf{b}_{-i}) \succ \nu(i)$ then by Lemma \ref{lem:opt-gsp-upperboundprice}, we have that:
	\begin{align*}
	u_i(b_i', \mathbf{b}_{-i}) &= \delta_{\tau(i)}^{\sigma(i,b_i', \mathbf{b}_{-i})}  v_i - p_i \geq \delta_{\tau(i)}^{\sigma(i,b_i', \mathbf{b}_{-i})} v_i- \delta_{\tau(i)}^{\sigma(i,b_i', \mathbf{b}_{-i})} \frac{v_i}{2} = \delta_{\tau(i)}^{\sigma(i,b_i', \mathbf{b}_{-i})} \frac{v_i}{2} \geq \delta_{\tau(i)}^{\nu(i)} \frac{v_i}{2}
	\end{align*}
	Otherwise, suppose deviating to $b_i'$ gets $i$ a slot $\sigma(i, b_i', \mathbf{b}_{-i}) \prec \nu(i)$. 
	Then since the allocation algorithm maximizes (apparent) welfare and allocating $i$ to $\nu(i)$ was feasible, it must be that:
	
	\begin{align*}
	\delta_{\tau(i)}^{\sigma(i,b_i', \mathbf{b}_{-i})}\frac{v_i}{2} + \sum_{j\neq i} \delta_{\tau(j)}^{\sigma(j,b_i', \mathbf{b}_{-i})} \mathbf{b}_j \geq \delta_{\tau(i)}^{\nu(i)} \frac{v_i}{2} + \sum_{j\neq i} \delta_{\tau(j)}^{\nu(j)}\mathbf{b}_j
	\end{align*}
	
Here, the summation on the left-hand side is the apparent welfare (excluding $i$) given the allocation selected under deviation ($b_i', \mathbf{b}_{-i}$); we will write this quantity as as $\widehat{W}^{-i}(b_i', \mathbf{b}_{-i})$. The summation on the right-hand side is what the apparent welfare  (excluding $i$) would be if the (truly optimal) assignment $\mathbf{\nu}$ had been chosen instead; we will write this as $\widehat{W}_{\boldsymbol{\nu}}^{-i}$. Then we write:
	\begin{align*}
	\delta_{\tau(i)}^{\sigma(i,b_i', \mathbf{b}_{-i})} \frac{v_i}{2} \geq \delta_{\tau(i)}^{\nu(i)} \frac{v_i}{2} + \widehat{W}_{\mathbf{\nu}}^{-i} - \widehat{W}^{-i}(b_i', \mathbf{b}_{-i}).
	\end{align*}

	As Lemma 4 guarantees that the undiscounted price cannot be more than the bid, we have:
		\begin{align*}
		u_i(b_i',\mathbf{b}_{-i})= \delta_{\tau(i)}^{\sigma(i,b_i', \mathbf{b}_{-i})} v_i - p_i \geq \ \delta_{\tau(i)}^{\sigma(i,b_i', \mathbf{b}_{-i})} \frac{v_i}{2} \geq  \delta_{\tau(i)}^{\nu(i)} \frac{v_i}{2} + \widehat{W}_{\mathbf{\nu}}^{-i} - \widehat{W}^{-i}( b_i', \mathbf{b}_{-i}). 
	\end{align*}
	We can drop $\widehat{W}_{\mathbf{\nu}}^{-i}$ and still have a true inequality, so we focus on how different $\widehat{W}^{-i}(b_i', \mathbf{b}_{-i})$ can be from ${W}^{-i}(\mathbf{b})$. And since we assume conservative bids, we must have $\widehat{W}^{-i}( b_i', \mathbf{b}_{-i}) \leq  W^{-i}( b_i',\mathbf{b}_{-i})$. Hence, we can rewrite the inequality we have as:
	\begin{align*}
	u_i(b_i', \mathbf{b}_{-i}) \geq \delta_{\tau(i)}^{\nu(i)}\frac{v_i}{2}  - W^{-i}( b_i', \mathbf{b}_{-i})
	\end{align*}
	Now, we need to bound $W^{-i}(b_i', \mathbf{b}_{-i})$ in terms of $W(\mathbf{b})$. We will do this very coarsely. Notice that in any allocation, the algorithm will always fill all the slots. Let $\delta^{\max}$ be the maximum discount rate in the first slot - that is, $\max_{j} \delta_j^{1}$ - and let $\delta_{\min}$ be the minimum discount rate for the \emph{last} slot (i.e. $\min_{j} \delta_j^{n}$). By monotonicity and full allocation, we know then that at the very most, we have:
	\begin{align*}
	W^{-i}(b_i', \mathbf{b}_{-i}) \leq \sum_{j\neq i} \delta^{\max}v_j = \delta^{\max} \sum_{j \neq i} v_j 
	\end{align*} 
	and at the very least, we have:
	\begin{align*}
	W( \mathbf{b}) \geq \sum_{j} \delta^{\min} v_j = \delta^{\min} \sum_{j} v_j.
	\end{align*}
	
	But that means that whatever $W( \mathbf{b})$ is, we must have that:
	\begin{align*}
	W^{-i}(b_i', \mathbf{b}_{-i}) \leq \frac{\delta^{\max}}{\delta^{\min}} \cdot \frac{\sum_{j\neq i} v_j}{\sum_{j}v_j} W(\mathbf{b}).
	\end{align*}
	(To see this, just note that $ 1\leq W(\mathbf{b})/ (\delta^{\min} \sum_j v_j)$ and multiply the inequality with $W^{-i}$ by 1 and apply this inequality to $\delta^{\max} \sum_{j\neq i} v_j \cdot 1$.)

	But now, using this upper bound for $W^{-i}$ to upper bound the negative term in the inequality above, we can write that
	\begin{align}\label{ineq:opt-gsp-smooth}
	u_i(b_i',\mathbf{b}_{-i})  \geq {\delta^{\nu(i)}}\frac{v_i}{2} - \frac{\delta^{\max}}{\delta^{\min}} \frac{\sum_{j \neq i} v_j}{\sum_{j} v_j} W(\mathbf{b}).
	\end{align}
	Inequality \ref{ineq:opt-gsp-smooth} thus holds in the case that $i$ gets a worse slot than $\nu(i)$ under the deviation, but of course it also holds true in the case that $i$ gets a better slot. Thus, it always holds, so we can sum over bidders to write:
	\begin{align*}
	\sum_{i} u_i(b_i', \mathbf{b}_{-i}) &\geq \sum_{i} \delta^{\nu(i)} \frac{v_i}{2} - \frac{\delta^{\max}}{\delta^{\min}} \sum_i \frac{\sum_{j\neq i} v_j}{\sum_j v_j} W(\mathbf{b}) = \frac{\text{OPT}}{2} - \frac{\delta^{\max}}{\delta^{\min}} W(\mathbf{b}) \frac{1}{\sum_{j} v_j} \sum_i \sum_{j\neq i} v_j \\
	&= \frac{\text{OPT}}{2} - \frac{\delta^{\max}}{\delta^{\min}} W(\mathbf{b}) \frac{1}{\sum_j v_j} \cdot (n-1) \sum_{j} v_j
	\end{align*}
	where the last inequality follows since each agent's valuation appears exactly $n-1$ times over the double sum.
	But then we have that
	\begin{align*}
	\sum_i u_i(b_i', \mathbf{b}_{-i}) \geq \frac{\text{OPT}}{2} - \frac{\delta^{\max}}{\delta^{\min}} (n-1) W(\mathbf{b}).
	\end{align*}
	
	Thus, this game is $(\frac{1}{2}, \frac{\delta^{\max}}{\delta^{\min}}(n-1))$-semismooth.
\end{proof}

\begin{cor}
	The (Opt,GSP) mechanism has an instance-specific upper bound on Price of Anarchy of:
	\begin{align*}
	\text{PoA} \leq 2 + 2 (n-1) \frac{\delta^{\max}}{\delta^{\min}}.
	\end{align*} If we assume that there are $m$ slots and all discount curves are geometric and strictly ordered (e.g. $c_\tau = c_{\tau'}$ and $\delta_{\tau_1} \geq \delta_{\tau_2} \geq...\geq \delta_{\tau_k}$ for some $k$), then an upper bound is given by:
	\begin{align*}
	2+ 2\cdot(n-1) \frac{\delta_{\tau_1}}{\delta_{\tau_k}^m}
\end{align*}
\end{cor}
We remark that this bound is potentially exponential in the number of bidders in the case of \emph{geometric} discount curves, but linear in the case of \emph{linear} discount curves (assuming a fixed set of discount curves). And while this bound is likely too pessimistic, we can give a lower bound as well:
 \begin{theorem}\label{thm:optgspexample}
 	Let $(\A,\P_{\A}) = $ (Opt,GSP). Then there exists a conservative 3-bidder 3-slot example that gets a competitive ratio arbitrarily close to 3/4.
\end{theorem}

\begin{proofsketch}
	Again, we construct a counterexample, prove it is an equilibrium, and optimize the welfare subject to equilibrium conditions. Here, a two-player two-slot example cannot suffer a high PoA, because the inefficient assignment of bidders would allow for a profitable deviation of the bidder in the worse slot (or the better slot if the price were too high). But with three bidders and three slots, we can find an example where two of the bidders effectively exert a ``joint" externality, and no single bidder has any incentive to deviate despite the allocation being suboptimal overall.
\end{proofsketch}
	
	Of the results we have derived, this mechanism has the least-tight upper bound on the price of anarchy, and the weakest lower bound. On the other hand, the following intuition suggests that the mechanism should perform relatively well: by construction, whenever the mechanism has access to the true valuations, its allocation is optimal. It does not have access to true valuations because it is not incentive-compatible, but GSP, like VCG, does somewhat ``protect'' a bidder from the risk of overpaying. Thus bidders may have less incentive to greatly shade their bid. We leave formalizing and exploring this intuition and improving these PoA bounds to future work.

%% file: content/eq-ec.tex
In this section, we provide the first analytical characterization of Bayes-Nash equilibrium in the two-slot, two-bidder case with ad types and under the assumption that bidder values are drawn independently and from identical uniform distributions over the interval $[0,1]$. In particular, we show the existence of simple equilibria that are symmetric in form and mostly natural. To find these equilibria, one may assume as a heuristic that an equilibrium exists, and derive first-order conditions; while this is a natural way to do so, ultimately, the proof is easiest when positing the existence of a linear equilibrium and verifying that the prescribed strategies are, in fact, best responses to one another. That is the approach we will take here. As mentioned, we defer details of some proofs to an online extended version of this paper in favor of proof sketches.

For each auction type, we assume there are two slots, two discount types A and B, and one bidder of each type. We assume that the discount types have the form $(1,\delta_A)$ for type A and $(1,\delta_B)$ for type B; i.e., geometric discount curves that both have a constant factor of 1. (This assumption can be easily relaxed at the cost of carrying around some extra notation.) Throughout, we will assume without loss of generality that $\delta_A < \delta_B$, and define $\Delta := \frac{1-\delta_B}{1-\delta_A}< 1$.  For the sake of efficiency, we say a bidder `wins' if they win the first slot.

\begin{table}[h]
	\centering
	\begin{tabular}{ccc}
		& GSP & VCG\\
		\toprule
		Greedy &$(1-\delta_A) v_A, (1-\delta_B)v_B$ & $\frac{1-\delta_A}{1-\delta_B} v_A, \frac{1-\delta_B}{1-\delta_A} v_B$\\
		Opt & $(1-\delta_A) v_A, (1-\delta_B)v_B$ & $(v_A, v_B)$
	\end{tabular}
	\caption{2 bidder, 2 type case, simple equilibrium strategies \label{tab:strats}}
\end{table}

Table \ref{tab:strats} displays the simple linear equilibria we discover. These equilibria are unique among linear equilibria (but not in general). Notice that in each setting, player strategies are symmetric up to relabeling. In other words, the \emph{form} of the strategy is symmetric, despite the fact that the particular strategy will differ due to different discount rates. Also, other than the VCG mechanism, each auction involves some shading. For GSP pricing, the downward shading coincides with each bidder's marginal benefit of the first slot relative to the second. But when VCG pricing is combined with greedy allocation, Bidder B shades down while Bidder A shades \emph{up}\footnote{Whilte this may be counterintuitive, note that with greedy allocation, bidding higher increases the win probability, and under GSP pricing, bidding higher does not (directly) increase the price paid. However, overbidding results in the possibility of winning at a price higher than one's valuation.}.

\begin{table}[h]
	\centering
	\begin{tabular}{ccc}
		& GSP & VCG\\
		\toprule
		Greedy &$\frac{1-\delta_A}{6} \Delta^2 + \frac{1-\delta_B}{6}(3-2\Delta) $& $\frac{1-\delta_A}{6} \Delta^3 + \frac{1-\delta_B}{6} \Delta\left(3-2\Delta^2\right)$\\
		Opt & $\frac{1-\delta_A}{6} \Delta^3 + \frac{1-\delta_B}{6} \Delta\left(3-2\Delta^2\right)$& $\frac{1-\delta_A}{6} \Delta^2 + \frac{1-\delta_B}{6}(3-2\Delta) $
	\end{tabular}
	\caption{2 bidder, 2 type case, equilibrium revenue \label{tab:revs}}
\end{table}

Table \ref{tab:revs} gives the expected revenue for each of the equilibria described in Table \ref{tab:strats}. As with Table \ref{tab:strats}, several features are noteworthy. First, immediately we can see that both the two standard formats, as well as the two nonstandard formats, are (expected) revenue equivalent. This may be surprising, given the variation in payment rules and strategies; however, we will see that the strategies are such that the win condition and payment conditional on winning work out to be the same. Second, we note that as expected, if we allow $\delta_A=\delta_B=\delta$, we recover the equivalent revenue to the VCG mechanism for all four auction formats. This is because when discounts are the same and there are only two slots, the greedy allocation is equivalent to the optimal allocation, and GSP pricing coincides with externality pricing, so the matrix of auction formats collapses to a single row and column. Moreover, if we set $\delta=0$, we recover the revenue of the standard second price auction with two uniform bidders (which is sensible, because if $\delta=0$, the auction is effectively simply a second price auction for the only slot with any value). Finally, we note that it is not immediately obvious whether revenue increases or decreases with discount values (since $\Delta$ is a function of $\delta_A,\delta_B$); again, it is easy enough, if uninspiring, to take the derivative and find that revenue decreases as either discount factor increases. It may be surprising that revenue decreases when bidders can derive more total welfare, but the principle is easy to see in the extreme: if there is no difference in clickthrough rates, bidders need not bid high at all\footnote{We assume there is no reserve; we leave as an open problem questions around designing optimal auctions with ad types.}, as they may as well take the second slot.  

These revenue results let us make \emph{equilibrium}, rather than fixed bid\footnote{For instance, it is known in the standard position setting that GSP prices are lower bounded by VCG prices for any fixed set of bids, but such a statement makes no prediction when bidders adjust their strategies to equilibrium.}, comparisons of revenue. In particular, simple, if involved, algebra allows us to proclaim the following relationship between revenue:

\begin{restatable}[Equilibrium Revenue]{theorem}{eqrev} \label{thm:eqrev}
  Consider a two-bidder, two-type, two-slot setting with bidder valuations drawn from a standard uniform distribution. Then in simple, linear equilibria: 
  \begin{align*}
  \mathcal{R}_{\opt}^{\vcg} = \mathcal{R}_{\greedy}^{\gsp} \geq \mathcal{R}_{\greedy}^{\vcg} = \mathcal{R}_{\opt}^{\gsp}
  \end{align*}
\end{restatable}

Importantly, these results only apply to our simple setting; it is unclear whether the revenue, welfare, or other predictions carry over into a general setting. 
\longversion{And indeed, in Section \ref{sec:nonlinear}, we show that one of the least extensive generalizations does not admit such an analytically tractable characterization.} While it is possible that more complicated analytic equilibrium may exist, 
 it is difficult to foresee how such an equilibrium might be found. Moreover, it is possible that equilibrium strategies, even if they do exist, are complicated to calculate and implement. Thus, in Section \ref{s:exp}, we turn our attention to empirical study of revenue under realistic bid distributions, where (coarse correlated) equilibria are \emph{learned} via no-regret learning techniques. 
\subsection{Greedy GSP}
	In this setting, the higher bidder gets the top slot at a price of the lower bid, and the lower bidder gets the bottom slot at a price of 0.  We obtain the following theorem:
\begin{theorem}\label{thm:eqggsp} Suppose that $(\A, \P_{\A})$ are (Greedy, GSP). Then in the two slot, two bidder, uniform case, the strategy profile
	\begin{align*}
	\left(b_A(v_A),b_B(v_B) \right) \coloneqq ((1-\delta_A)v_A, (1-\delta_B)v_B)
	\end{align*}
	 is a Bayes-Nash equilibrium. Among conservative linear equilibria, it is unique.
\end{theorem}

To prove this theorem, and all our other equilibrium claims, we must show that the strategies are a best response to each other under the distribution of bidder valuations. The easiest way to do so is to take an ex-interim perspective for each bidder, assume the opposing bidder uses the claimed strategy, and allow the original bidder to optimize freely. Then, one shows that maximum is achieved at exactly the value prescribed by strategy. Since said strategy prescribes an ex-interim best-response at every possible valuation, it is a best-response. 

To show that prescribed strategy is in fact a best-response, we decompose each bidder's expected payoff into the sum of their expected profit if they win (which can be further decomposed into the probability of winning times the expected profit given they win) and their expected profit if they obtain the worse slot. Viewing this payoff as a function of the players' bid, we maximize that function, and show the optimal bid is exactly the prescribed strategy for each player. 		

We will prove the case of $(Greedy, GSP)$ in detail; for the other cases, we provide a proof sketch and defer the more involved proofs to the appendix. 

\begin{proof}
	Consider Bidder A's perspective after she learns her valuation $v_A$. If A wins, she pays $b_B$ and gets value $v_A$; if she loses, she gets $\delta_A v_A$ and pays nothing. Then:
	\begin{multline}\label{eqn:g-gsp-exinterim}
	\E_{v_B \sim U[0,1]}\left[u_A|b_A\right] = (v_A-\E[b_B|b_B <b_A])  \Pr[b_B< b_A]  + \delta_A v_A (1-\Pr[b_B < b_A])
 	\end{multline}
 	Since we wish to show that $(1-\delta_A)v_A$ is a best-response to $(1-\delta_B)v_B$, we can assume that $b_B= (1-\delta_B)v_B$. Hence, A wins if and only if $v_B < b_A/(1-\delta_B)$. Under the uniform distribution, $\Pr[x<c]=\min\{c,1\}$ and $\E[x|x<c] = \frac{\min\{c,1\}}{2}$ for $c>0$.  Thus we can apply these to Equation \ref{eqn:g-gsp-exinterim} to write:
 	\begin{align}
 		\begin{split}\label{eq:ut}
 	\E_{v_B \sim U[0,1]}\left[u_A|b_A\right] &= \left(v_A - (1-\delta_B)\frac{b_A}{2(1-\delta_B)}\right)\frac{b_A}{1-\delta_B} + \delta_A v_A(1-\frac{b_A}{1-\delta_B})  \\
 	&= \frac{v_A b_A}{1-\delta_B} - \frac{b_A^2}{2(1-\delta_B)} - \frac{\delta_A b_A v_A}{1-\delta_B} + \delta_Av_A
 	\end{split}
 	\end{align} 
 whenever $b_A \leq (1-\delta_B)$, and $u_A = v_A - \frac{1-\delta_B}{2}$ otherwise. In other words, $A$'s payoff will be either the left-hand side of \Cref{eq:ut}, which we denote as $u_A(b_A)$ for brevity, or the ``cap'' of $v_A -\frac{1-\delta_B}{2}$, depending on whether $b_A$ is less or more than $(1-\delta_B)$. So in principle, we need to find the bid that maximizes $u_A$ on $[0, 1-\delta_B]$, and then check whether or not it gives a better payoff than the cap. But notice that $u_A(1-\delta_B)=v_A - \frac{1-\delta_B}{2}$, and increasing $b_A$ beyond $1-\delta_B$ cannot improve payoff, so it suffices to simply find the maximum of $u_A$ over $[0,1-\delta_B]$.
 
  Notice that $u_A(b_A$) is continuous and differentiable in $b_A$ on $[0,1-\delta_B]$. The first and second derivatives of $u_A$ are: 
 \begin{align*}
 	u_A'(b_A) = \frac{v_A(1-\delta_A)}{1-\delta_B} - \frac{b_A}{1-\delta_A}  \qquad u_A''(b_A) = -\frac{1}{1-\delta_A}
 \end{align*}
Hence $u_A$ is strictly concave. Suppose that $b_A^* \coloneqq (1-\delta_A)v_A < 1-\delta_B$. Then $b_A^*$ satifies the first order condition and so is a global maximum. On the other hand, if $b_A^* \geq (1-\delta_B)$, then because $u_A$ is increasing right up until $(1-\delta_B)$, $u_A$ takes it maximum at $b_A = (1-\delta_B)$. But, bidding $(1-\delta_A)v_A$ results in the same payoff as bidding $1-\delta_B$ (because of the ``cap''). Thus, regardless of what $v_A$ is, the strategy $b_A^*= (1-\delta_A) v_A$ is a best-response if B is bidding $(1-\delta_B) v_B$. 
 	Reversing roles and considering B's perspective gives exactly the same logic. Hence, the pair of strategies form an equilibrium. To see uniqueness among linear equilibria, notice that as long as $b_B$ is linear, i.e. $b_B(v_B) =\beta v_B$ for some fixed $0 \leq \beta \leq 1 $, the form of Equation \ref{eq:ut} holds, and the particular choice of $\beta$ cancels out just as it did for $(1-\delta_B)$; again, then, the optimal bid will be $(1-\delta_A)v_A$. A similar argument holds for B.
\end{proof}

\begin{proposition}\label{prop:gspgreedyrev}Under the linear equilibrium described above, with $\delta_A <\delta_B$, we have that the expected revenue is given by:
	\begin{align*}
	\E[\mathcal{R}] =  \frac{(1-\delta_A)\Delta^2}{6} + \frac{1-\delta_B}{6}\left(3-2 \Delta\right)
	\end{align*}\end{proposition}
Before we sketch the proof, note that if we let $\delta_A=\delta_B=0$, we immediately recover $\frac{1}{3}$, which is the revenue of the standard second price auction with two bidders drawn from $U[0,1]$. Second, if we let $\delta_A=\delta_B=\delta$, then $\Delta =1$, and we see that $\E[\mathcal{R}] ={{(1-\delta)}\over{3}}$. That is, revenue decays \textit{linearly} to that of the standard second price auction as $\delta \to 1$. 
\longversion{:
\begin{align*}
\E[R] &= \frac{\left((1-\delta)(1)^2\right)}{6} + \frac{1-\delta}{6} (3-2)= \frac{1}{3} (1-\delta).
\end{align*}
}

\shortversion{\begin{proofsketch}
		Once one identifies the revenue curve over the joint distribution of valuations, this is a straightforward calculation. To identify the revenue curve, consider the case when A wins the top slot. Then it must be that $b_A \geq b_B$, which, under the equilibrium strategies, holds if and only if $b_A \geq \Delta b_B$. If A wins, then she pays $b_B$, which is $(1-\delta_B)v_B$. Similarly, B wins when $b_A \leq b_B$ and pays $(1-\delta_A) v_A$. Then, we set up a double integral over the uniform distribution for $v_A$ and $v_B$, but split the inner integral up, as it will have a different integrand from $0$ to $\Delta v_B$ (where B wins) and $\Delta v_B$ to $1$ (where A wins). Evaluating the double integral, the result follows. 
\end{proofsketch}}
\longversion{
\begin{proof}
	A wins if $b_A\geq b_B$, which happens when:
	\begin{align*}
	b_A \geq b_B \iff (1-\delta_A) v_A \geq (1-\delta_B)v_B \iff v_A \geq \Delta v_B
	\end{align*}
	If A wins, she pays $b_B$, and so the revenue is $b_B = (1-\delta_B)v_B$; otherwise, it is $b_A = (1-\delta_A) v_A$. Thus we can write the expected revenue as:
	\begin{align*}
	\E[R] &= \int_0^1 \int_0^1 R(v_A,v_B) dP(v_A)dP(v_B)\\
	&=\int_0^1 \int_0^{\Delta v_B} (1-\delta_A) v_A dv_A  dv_B + \int_0^1 \int_{\Delta v_B}^1 (1-\delta_B) v_B dv_A dv_B\\
	& = (1-\delta_A) \int_0^1 \frac{v_A^2}{2} \biggr|_0^{\Delta v_B} + (1-\delta_B) \int_0^1  v_B v_A\biggr|_{\Delta v_B}^1 dv_B\\
	&=(1-\delta_A) \int_0^1 \Delta^2 \frac{v_B^2}{2} dv_B + (1-\delta_B) \int_0^1 v_B - \Delta v_B^2 dv_B\\
	& = (1-\delta_A) \Delta^2\frac{v_B^3}{6} \biggr|_0^1 + (1-\delta_B) \left[\frac{v_B^2}{2} - \Delta \frac{v_B^3}{3}\right]\biggr|_0^1\\
	& = \frac{(1-\delta_A)\Delta^2}{6} + \frac{1-\delta_B}{6}\left(3-2 \Delta\right)
	\end{align*}
\end{proof}
}
\subsection{Optimal Allocation and GSP Pricing}
In this setting, the auctioneer chooses between the allocation $(A,B)$ and $(B,A)$. Note that:
\begin{align*}
(A,B) \succeq (B,A) \iff b_A + (1-\delta_B)b_B \geq b_B +(1-\delta_A) v_A \iff b_A \geq \Delta b_B
\end{align*}
Suppose bidder A is the winner. Then A is charged the smallest bid $b$ such that
$b \geq \Delta b_{B}$, which is just $\Delta b_{B}$. Similarly, if B wins, he will be charged $b_A/\Delta$.

\begin{theorem} \label{thm:eqoptgsp} Suppose that $(\A, \P_{\A})$ are (Opt, GSP). Then in the two slot, two bidder, uniform case, the strategy profile
	\begin{align*}
	\left(b_A(v_A), b_B(v_B)\right) \coloneqq ((1-\delta_A)v_A, (1-\delta_B)v_B)
	\end{align*}
	is a Bayes-Nash equilibrium.
\end{theorem}
\begin{proofsketch}
	As in Theorem \ref{thm:eqggsp}, we show that each strategy is a best-response to the other, and take particular care with the piecewise-nature of the payoff.
\end{proofsketch}

\begin{proposition}[(Opt,GSP) Revenue]\label{prop:optgsprev}
	Under the linear equilibrium described above, with $\delta_A <\delta_B$, we have that the expected revenue is given by:
	\begin{align*}
	\E[\mathcal{R}] =\frac{\Delta^3(1-\delta_A)}{6} + (1-\delta_B) \Delta \left[\frac{1}{2} - \frac{\Delta^2}{3}\right]
	\end{align*}
\end{proposition}
\begin{proofsketch}
		The proof follows the same structure as that of Theorem \ref{thm:eqrev}. Again, once one identifies the relevant events and payoffs, the calculation is a straightforward double integral. A wins whenever $b_A \geq \Delta b_B$; thus given the equilibrium strategies, A wins whenever $v_A \geq \Delta^2 v_B$. The payment A makes if she wins is the smallest payment $p$ such that $p \geq \Delta b_B$, which is exactly $\Delta b_B$; under the equilibrium, then A will pay $\Delta (1-\delta_B) v_B$. Similarly, $B$ wins whenever $v_A \leq \Delta^2 v_B$, and pays $v_A \cdot (1-\delta_A)/\Delta$. Setting up the integral in pieces as before and evaluating yields the claim. 
\end{proofsketch}

\subsection{Greedy VCG}
Now we turn to $(Greedy,VCG)$. Here, $A$ wins whenever $b_A \geq b_B$,  but the pricing is VCG; that is, if $A$ wins, she pays $(1-\delta_B) b_B$. In this case, we again find a simple linear equilibrium; however, in this setting, the equilibrium given is not unique.

\begin{theorem} \label{thm:eqgreedygsp}Suppose that $(\A, \P_{\A})$ are $(Greedy,GSP)$. Then in the two slot, two bidder, uniform valuation case, the strategy profile:
	
	\begin{align*}
		(b_A(v_A), b_B(v_B)) \coloneq \left({(1-\delta_A)\over (1-\delta_B)} v_A, {(1-\delta_B)\over(1-\delta_A)} v_B\right)= \left(\frac{v_A}{\Delta}, \Delta v_B  \right)
		\end{align*}
	is a Bayes-Nash equilibrium.
	\end{theorem}
Notice that since $\Delta < 1$, Bidder A overbids while Bidder B shades down. The intuition for this structure is that A has a higher expected (marginal) effective valuation for the first slot than B; thus greedy allocation encourages overbidding on A's parts to increase win probability without a strong enough countervailing check via B's bid. For B, by contrast, the fact that A is overbidding and has a higher marginal valuation anyway makes the it possible to achieve negative payoff even if bidding only his valuation. 

\begin{proofsketch}
	Once again, we show that each strategy is a best-response to the other, and take particular care with the piecewise-nature of the payoff.
\end{proofsketch}

\begin{proposition}\label{thm:revgreedyvcg}
	In this equilibrium above, revenue is given by:
	\begin{align*}
	\E[\mathcal{R}(v_A,v_B)] =\frac{\Delta^3(1-\delta_A)}{6} + (1-\delta_B) \Delta \left[\frac{1}{2} - \frac{\Delta^2}{3}\right]
	\end{align*}
\end{proposition}

Notice that this is the \emph{same} revenue as under (Opt,GSP). Why should this be? It turns out that the structure of the (Greedy, GSP) equilibrium implies the same win conditions, in terms of realized bidder valuations, and the same payments conditional on winning. Informally, the equilibrium strategies ``adjust'' for the differing pricing and allocations rules.

	\begin{proofsketch}Notice that A wins if her value is $b_A \geq b_B$, which under the strategy profile is true iff $v_A \geq \Delta^2 v_B$. If she wins, she pays $(1-\delta_B) \Delta v_B$. Similarly, B wins if $v_A \leq \Delta^2v_B$, and pays $v_A (1-\delta_A)/\Delta$ if he wins. Recall that in OPT + GSP, A won if $b_A \geq \Delta b_B$, but given the strategy profile, this is true whenever $(1-\delta_A) v_A \geq (1-\delta_B) \Delta v_B$. So A wins whenever $v_A \geq \Delta^2 v_B$. Similarly, under OPT + GSP, A paid $\Delta b_B$, which under the profile is $ \Delta (1-\delta_B)v_B$. A similar argument works for B. And thus, the calculation works out to be exactly the same.
\end{proofsketch}	

\subsection{Optimal Allocation and VCG Pricing} Recall that (Opt, VCG) is just the standard VCG mechanism, which is well-known to have a natural dominant strategy equilibrium in truthful bidding. Thus we need only calculate the revenue:

\begin{proposition}\label{thm:optvcgrev} In the truthful equilibrium of (Opt, VCG), revenue is given by $\frac{\Delta^2 (1-\delta_A)}{6} + \frac{1-\delta_B}{6}\left(3 - 2\Delta\right)$. 
	\end{proposition}

Notice that this is, perhaps surprisingly, the same revenue as the Greedy + GSP auction. As before, this is because the winning events and conditional payments are exactly the same in this format (in this setting) as under the linear equilibrium under Greedy + GSP. 

\begin{proofsketch} 
		Again, rather than recalculating this expected revenue, notice that A wins whenever $b_A \geq \Delta b_B$; since we are considering the dominant strategy truthful equilibrium, this is true if and only if $v_A \geq \Delta v_B$, which was the same win condition for A under Greedy + GSP when A's bid was $(1-\delta_A) v_A$ and B's was $(1-\delta_B) v_B$. As for payment, if A wins, she pays her externality, which , since bidders bid truthfully, is just $(1-\delta_B)v_B$. This again is the same payment as under Greedy+GSP when B bid $(1-\delta_B) v_B$. So again, the calculation follows in the same way. 
	
		\end{proofsketch}

\subsection{Revenue Comparison}
Now we compare the revenue of the different auction forms.  
Again, this is using the revenues calculated above and provided in table 2; that is, the simple linear equilibria and assuming that $\delta_A <\delta_B$. We stated the revenue hierarchy before as Theorem \ref{thm:eqrev}:

\eqrev*

\begin{proofsketch}The two inequalities follow by inspection of Table \ref{tab:revs}, so only the inequality needs proof. To do this, we simply expand out the difference between $R_{\greedy}^{*\gsp}$ and $R_{\greedy}^{*vcg}$. At this point, we just need to show that this difference is always positive; this can be easily seen by graphing the function, but we also analytically show that this holds in the appendix.\end{proofsketch}

\longversion{
\subsection{More complicated settings}\label{sec:nonlinear}
Unfortunately, though the two bidder case admits elegant linear equilibria, expanding the setup as simply as to two slots, two bidders of one type and one bidder of another immediately eliminates hope of finding a simple linear equilibrium in general. 
To see this, one can posit a linear equilibrium again symmetric up to discount types. Then beginning with the rare player, one can attempt to solve for this linear equilibrium, and one way to attack this is to view the game as a two-stage game for that player: first, there is a preliminary game in which the players bids determine who partcipiates in a second price auction and who sits unallocated entirely; then, the there is a continuation game for the selected players in which their (original) bid determines their result in the auction. By calculating the payoff of a given bid in this continuation game, one can easily write the payoff of a bid, and then it is easy to see that if the opposing type is playing a linear strategy, a linear strategy will not be optimal. 
}

%% file: content/exp-ec.tex
In this section, we test our theoretical predictions of revenue and PoA on both simulated and realistic data using \emph{No-Regret Learning} (NRL) algorithms to model bidders. These algorithms have guarantees of convergence to (more general notions of) equilibrium, and have also been proposed as potential solution concepts in their own right \cite{kleinberg2011beyond}. We have two key results. First, the theoretical Bayes-Nash equilibria are (approximately) discovered by NRL bidders. Second, we find evidence that the revenue relationships predicted by the theoretical analysis do appear in the data, but this can be sensitive to the particular valuations drawn, underscoring the need for experimentation in addition to theoretical analysis.

\paragraph{Approach.} 
 We outline here the core commonality across the experiments. Each is based on NRL algorithms, which converge\footnote{We describe the more formal meaning behind these statements  and the technical details more broadly in Section \ref{sec:appexperiments} of the online appendix.} to (Bayesian) \emph{coarse correlated equilibrium} (CCE) (see, e.g. \cite{roughgarden2016twenty})). 
 Though CCEs are more general than those we studied earlier in the paper, they may better reflect the real-world situation bidders face as they arise naturally from players learning to bid independently. 
 
 In each experiment and for each mechanism, we instantiate copies of the \emph{exponential weights} (EW) algorithm to represent each player. The players play in many repeated rounds, maintaining at every round a distribution over bids from a discrete bid space. Each round, bidders draw their bids from this distribution, and the mechanism uses the bids to compute an allocation and price for each bidder. We record the total revenue and welfare under the realized outcome; for each player and each alternative potential bid, we also calculate a \emph{counterfactual} outcome by re-running the mechanism with all \emph{other} players' choices held fixed. The player then observes his payoff under the counterfactual outcome and updates his distributions over actions accordingly.

 \paragraph{Experiment 1.} Our first experiment approximates the two-bidder, two-slot, uniform distribution case we analyzed in Section \ref{s:eq}.  
  For this experiment, we adopt the \emph{population interpretation} of Bayesian games, and so discretize the uniform valuation distributions into discrete uniform distributions over players with \emph{fixed}, evenly-spaced valuations. Each round, nature selects a single player from each population, and each player maintains their own strategy. This approach is analyzed in \cite{hartline2015no}; we modify the approach by adding an extensive random exploration period. This modification is inspired by \cite{feng2021}, which shows that including a sufficiently long exploration period in natural auction settings allows bidders to {provably} converge to \emph{specific} and natural Nash equilibria. In our case, we find that for each mechanism, we observe a very close correspondence between realized and theoretical bid distributions. Figure \ref{fig:exp1bcce}) displays the predicted bids (as a dashed line) and average observed bid for each valuation.
 
\paragraph{Realistic Data for Experiments 2 and 3.}  We use real data from an online platform with a large advertising business
and sample real bids to generate realistic valuation data. These are not \emph{literal} valuation data for two reasons. First, bidders on the platform face a more complicated setting than modeled, e.g., bidders compete in multiple sequential and simultaneous auctions, so bids may not precisely correspond to values. Second, we have normalized the data to protect the privacy of the participants. Hence, these bids are a reasonable proxy for real-world distributions, but may not be exactly such in practice.  

The first dataset we collect is the \emph{Random Advertisers} dataset, in which we sample 10 random advertisers who had between 100,000 and 200,000 impressions on a particular outlet and a day\footnote{Mobile Advertising, September 19, 2020}. For these 10 advertisers, we select 100,000 bids and normalize each advertiser's bids to fall within the unit interval and clamp at the 5th and 95th percentiles.  We can use this dataset to sample independently drawn valuations. The second is the \emph{Random Auction} dataset: we again fix the outlet and day and randomly select 100,000 auctions, this time normalizing each auction individually. The Random Auction dataset thus maintain correlation between bidders' valuations in a given auction, which may be, in some cases, an important real-world feature of the domain.

 \begin{figure}[h]
 	\centering
 	\begin{subfigure}[b]{0.4\textwidth} 
 		\centering
 				\includegraphics[width=\linewidth]{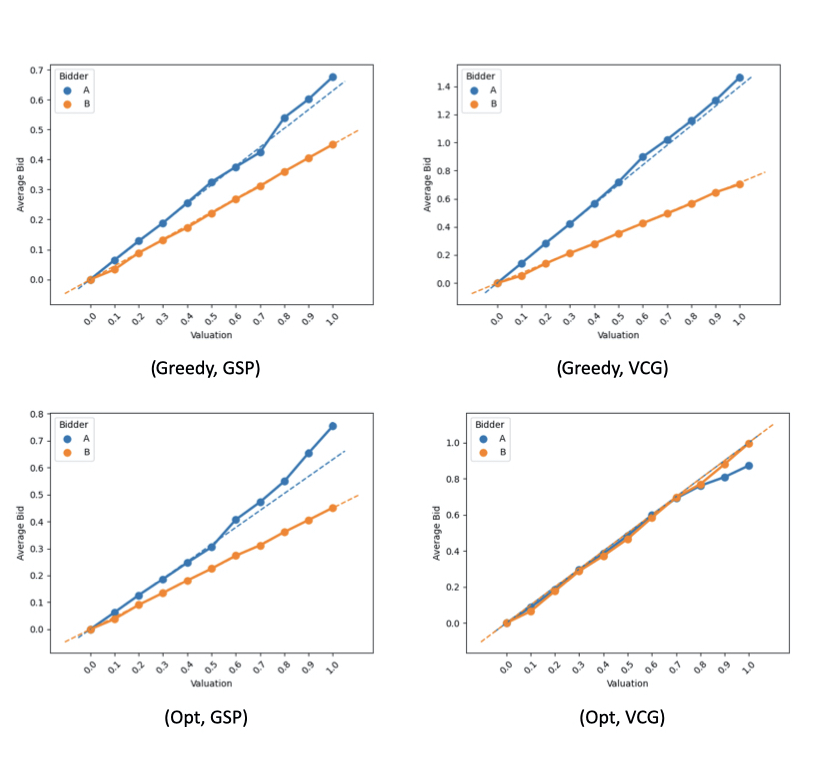}
 		\caption{Experiment 1: Mean bid by value. Dashes indicate theory. \label{fig:exp1bcce}}	
 	\end{subfigure}
 	 	 	\begin{subfigure}[b]{0.24\textwidth}
 	 		\centering
 	 		\includegraphics[width=\linewidth]{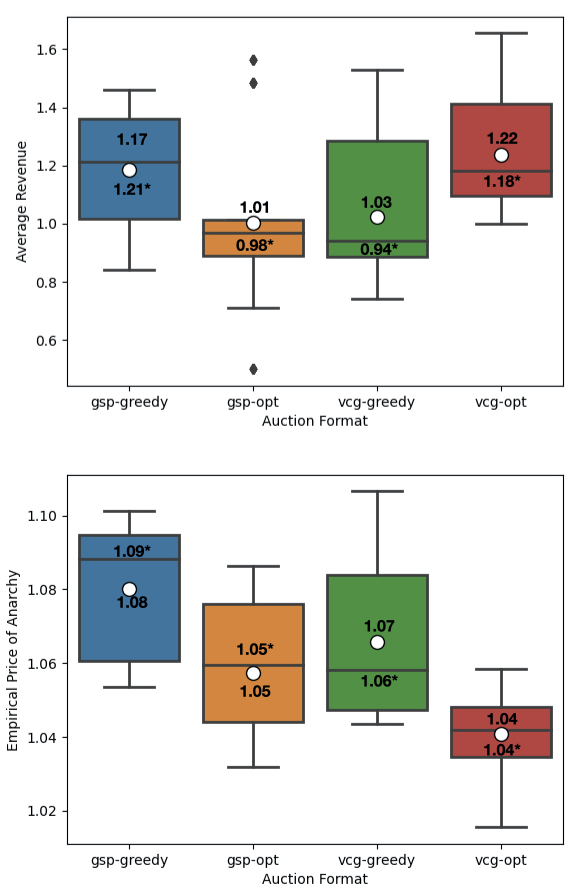}
 	 		\caption{Experiment 2: Fixed Valuations}		\label{fig:exp2}
 	 	\end{subfigure}
  	 	\begin{subfigure}[b]{0.24\textwidth}
  		\centering
  		\includegraphics[width=\linewidth]{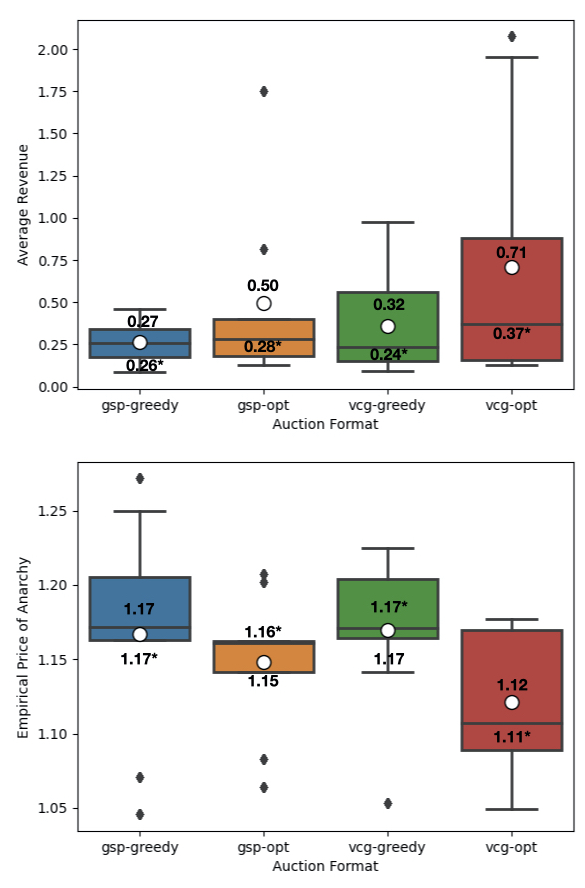}
  		\caption{Experiment 3: Random Valuations}		\label{fig:exp3}
  	\end{subfigure}
 	\caption{Experimental Results. Scales differ due to normalization. Means (medians) are marked by circles (lines).}
 \end{figure}

\paragraph{Experiments 2 and 3.} 
In Experiment 2, we use the Random Advertisers dataset. A protocol for a single round is as follows. We initialize an auction with 4 slots and 9 bidders of varying\footnote{We provide these and other implementation details in Section \ref{s:app-expdetails} of the online appendix.} geometric discount factors (each with a fixed constant multiplier of $1$). Each bidder has a valuation drawn independently from the Random Advertisers dataset, and is initialized with a fresh exponential weights algorithm over the (discretized) bidspace up to their valuation, as we would expect bids to be conservative in practice. After an exploration period, the players update via EW for 100 rounds, updating their bid distributions based on the realized and counterfactual bids each round.
Then, we sample a bid profile from the time-averaged joint distribution by uniformly selecting a time period and drawing a bid profile from the EW distributions of that given round; we average the revenue and welfare of 200 such samples as an estimate of the revenue and welfare given that valuation profile. For each auction format, we repeat this entire process for 200 total valuation profile draws, and average these together as an estimate of the format's revenue and welfare under the valuation distribution. We also compute the optimal allocation for each draw, and take the ratio of the average optimal value to the average welfare as an estimate of an empirical analogue (i.e. not worst-case) to the PoA. We plot revenue 
and empirical PoA in Figure \ref{fig:exp2}. Our third experiment is very similar to our second, except for the valuation sampling. Rather than sample valuations independently from the Random Advertiser dataset, we sample \emph{auctions} randomly from the Random Auction dataset, and assign those bids to bidder valuations. Then we proceed as described in Experiment 1. We plot revenue 
and the Empircal PoA in Figure \ref{fig:exp3}.

\paragraph{Experiment 2 and 3 Results.}
  In Experiment 2 (Figure \ref{fig:exp2}), a (rough) analog to the revenue hierarchy of Table \ref{tab:revs} is apparent. Additionally, the Empirical PoAs of all formats are significantly better than the worst-case bounds in Table \ref{tab:poalower}. In Experiment 3 (Figure \ref{fig:exp3}), the PoA is somewhat worse than those of Experiment 2, but again relatively far from the worst-case bounds we proved. On the other hand, the revenue hierarchy appears significantly different, with both GSP mechanisms doing worse than VCG. Together, these results suggest that the relative quality of mechanisms can be highly sensitive to the underlying distributions\footnote{
  We must caveat all these results by noting that because the bid space inherently high-dimensional, more samples of valuation profiles may be required for fidelity to the true distribution. This may be particularly true for the correlated case.}. 
  
\begin{algorithm}
	\SetAlgoNoLine
	\caption{Protocol for Experiment 1
		\label{exp:one}}
		\For{$(\A, \P_{\A}) \in \{\text{[Greedy,Opt]} \times \text([GSP,VCG]) \}$}{
		\For{$t \in 1,2...,100$}{
		 Draw 10 valuations. Initialize bidders with values and fresh Exp. Weights.  
		\For{$t \in 1,2,...,1000$}{
		 Initialize and run a 5 slot auction using $(\A, \P_{\A})$
		 Draw bids and fix them. 
		\For{$i \in \mathcal{I}$}{
		\For{$b' \in \{\frac{1}{d}, \frac{2}{d}, ...,1\}$}{
		 Re-run auction with all other players' bids fixed, but player $i$ using $b'$.\\
		 Save payoff. \\
		
		 Update ExpWeights. 
	}
	}
		\For{$t' \in 1,2,...100$}{
		 Test for CCE. If not, re-run another 1000 learning steps.\\
		 Draw number uniformly at random from the total number of learning steps to become round number. \\
		 Draw bid for each bidder from time average of ExpWeights at that round.\\
		 Run auction with these bids\\
		 Save round info.   
}
}
}
}
\end{algorithm}

%% file: content/conc-ec.tex
In this paper, we obtain several theoretical and empirical results for the Ad Types setting. We leave several open directions. In terms of Price of Anarchy: while we provide constant upper and lower bounds on the Price of Anarchy under greedy allocation, there remains a gap between these bounds. More substantially, while we provide a constant lower bound on the Price of Anarchy under optimal allocation with VCG pricing, our upper bound is instance-dependent and likely quite pessimistic; resolving this with either a constant upper bound, or identifying a family of arbitrarily bad examples, would be helpful. In terms of equilibrium characterization, it would be useful to identify (or rule out) analytical solutions in more complicated settings. In terms of empirics: understanding how our results would change as various features of the setting change would be valuable. For instance, even our large setting is still relatively small compared to modern instances encountered in online advertising today. Second, our understanding of how revenue and welfare may vary with discount curves in practice is not yet systematic; theory suggests that bidders ought to bid less aggressively as their valuation of further slots increase, but how much less aggressively, and how this is affected by auction format, is unknown.

%% file: content/app-ec.tex
\section{Notation Table}
\mbox{}
\nomenclature[A,01]{$\tau_i$}{The ad type of Player i}
\nomenclature[A,02]{$v_i$}{Valuation of Player i }
\nomenclature[A,04]{$b_i$}{Bid of Player i}
\nomenclature[A,03]{$b_i(v_i)$}{Bid mapping of player $i$}

\nomenclature[B,01]{$\A$}{An allocation algorithm}
\nomenclature[B,02]{$\P_{\A}$}{A pricing algorithm that uses $\A$ as its allocation subroutine}
\nomenclature[B,03]{$\mathbf{b}$}{A bid vector}
\nomenclature[B,04]{$\mathbf{b}_i$}{The ith component of the bid vector $\mathbf{b}$}

\nomenclature[C,01]{$\opt$}{The optimal welfare achievable given valuations}
\nomenclature[C,02]{$\text{SW}$}{Social Welfare}
\nomenclature[C,03]{$W$}{(True) Welfare as function, or specific welfare value in context }
\nomenclature[C,04]{$\widehat{W}$}{(Public) Welfare, i.e. welfare if bids were valuations, or public welfare value in context}
\nomenclature[C,05]{${W}^{-i}$}{(True) Welfare excluding $i$. , or true welfare excluding $i$ in context}
\nomenclature[C,06]{$\widehat{W}^{-i}$}{(Public) Welfare excluding $i$, or public welfare excluding $i$value in context}

\nomenclature[D,01]{$\delta_A$, $\delta_B$}{Player A's, Player B's discount rate for the second slot}

\nomenclature[D,02]{$v_A$, $v_B$}{Player A's,Player B's valuation for a click}

\nomenclature[D,03]{$b_A$, $b_B$}{Player A's, Player B's bid}

\nomenclature[D,04]{$\mathcal{R}(\mathbf{v})$}{Revenue. Given equilibrium/mechanism, $\mathcal{R}$ is function of valuations (i.e. $v_A$,$v_B$ in Section \ref{s:eq}).}

\nomenclature[E,01]{$\pi_{\A}(s,\mathbf{b})$}{The player in slot $s$ when the bid vector is $\mathbf{b}$. When clear, we may omit $\mathbf{b}$ and $\A$. }
\nomenclature[E,02]{$\tau(i)$}{The type of Player i}
\nomenclature[E,03]{$\tau(\pi(s))$}{The type of the player in slot $s$}
\nomenclature[E,04]{$\sigma_{\A}(i,\mathbf{b})$}{The slot that player $i$ receives under allocation algorithm $\A$ when the bid vector is $\mathbf{b}$. When clear, we may omit $\mathbf{b}$ and $\A$. }
\nomenclature[E,05]{$\boldsymbol{\nu}$}{Optimal allocation assignment vector as measured by (true) welfare given valutations}
\nomenclature[E,06]{$\nu(i)$}{$i$'s slot under the optimal assignment}
\printnomenclature

\section{Experimental Framework and Parameters} \label{sec:appexperiments}
The exact theoretical results we provide are interesting, but are limited to simple settings. 
In order to evaluate the equilibrium revenue and welfare of the mechanisms studied in a more general and realistic setting, we need to devise a computational approach we can take to data.  

It is computationally hard (\cite{chen2009settling}) in general to directly compute a Nash equilibrium, but at the cost of considering a more general equilibrium concept, we can make progress via no-regret learning (NRL) algorithms. In particular, we apply the well-known fact that the empirical distribution of action profiles taken by players using NRL algorithms forms a \emph{coarse correlated equilibrium}.

We describe the theory behind this approach in further detail below, but first we summarize our simulation framework at a high level. We explore two settings. We call the first (presented second in the main body of the text) the \emph{fixed valuation setting}, because bidders are assumed to have fixed valuations given which they learn to bid. We call the second the \emph{random valuation setting}, because bidders are modeled as randomly drawing a valuation each round. These two settings capture different but potentially equally reasonable models of repeated auctions. The fixed valuation setting is a good model for repeated auctions in which bidders are stable and have a sense of their opponents' valuations; the random valuation setting better captures bidders who may compete in auctions against entirely different opponents, and be competing for users with different valuations. Despite this apparent difference, however, the approach we use to learn their equilibria is largely similar, as will become apparent.

\subsection{Theoretical Underpinnings}
A \emph{coarse correlated equilibrium} is a more general equilibrium concept than a Nash equilibrium:

\begin{definition}[Coarse Correlated Equilibrium]
	We say a distribution $\D$ over actions is a \emph{coarse correlated equilibrium} (CCE) if for every player $i$, and every action $a'$:
	\begin{align*}
		\E_{\mathbf {a} \sim \D}[u_i(\mathbf {a})] \geq \E_{\mathbf{a} \sim \D} [u_i(a',\mathbf {a}_{-i})] 
	\end{align*}
In other words, if a strategy profile is drawn from a distribution $\mathcal{D}$, it is in each player's interest to follow their own part of the prescribed strategy under the assumption that others also will. 

We say that $\mathcal{D}$ is an $\epsilon$-approximate CCE if
	\begin{align*}
	\E_{\mathbf {a} \sim \D}[u_i(\mathbf {a})] \geq \E_{\mathbf{a} \sim \D} [u_i(a',\mathbf {a}_{-i})] -\epsilon
\end{align*}
\end{definition}

Next, we describe the simplest version of NRL that is useful for our purposes:

\paragraph{No-Regret Learning Framework.}
In the No-Regret Learning (NRL) Framework, a decision-maker faces an online sequence of decision problems with a fixed action set $\mathscr{A}$ over a finite time horizon. For each round $t=1...T$, the player selects $a^t$ from a probability distribution over $\mathscr{A}$, which we denote by $\alpha^t(a)$; the \emph{losses} (equivalently, payoffs) for each action $a$ are realized as $u^t(a)$, and the player receives the loss of whatever action he selected, i.e. $u^t(a^t)$\footnote{To be consistent with the literature, we will assume without loss of generality that payoffs are bounded between $0$ and $1$. This is also consistent with our normalization.}. (In the learning context, $u^t(a)$ is an arbitrary sequence of loss vectors, but in the game setting we can think of $u^t(a)$ as instead being a fixed utility function that depends on the decision-maker's choice but also on the choice of all other agents. In other words, $u^t(a) = u(a, \mathbf{a}_{-i}^{t})$.)

The (external) \emph{regret} of an action sequence is the difference between the payoff of the best fixed action $a^*$ and the player's payoff over the sequence. That is:
\begin{align*}
	 R_T =  \sum_{t=1}^{T} u^t(a^*) - \sum_{t=1}^{T} u^t(a^t) \qquad a^* \in \argmax_{\mathscr{A}} \sum_{t=1}^{T} u^t(a^*)
	 \end{align*}
 If an online learning algorithm promises that for every $\epsilon >0$, there exists a $T$ such that whatever the sequence of losses:
 \begin{align*}
	\E\left[\frac{R_T}{T} \right]\leq \epsilon
 \end{align*}
where the expectation is over the randomness of the algorithm, we say it is a \emph{NRL} algorithm. Of course, there is a huge literature on variants and generalizations of the simple framework presented here, but all we really need is the following result\footnote{A similar result can be obtained with high probability over the \emph{realized} sequence. (\cite{roughgarden2016twenty}), but its guarantees require extra time steps to allow for concentration of realized regret around expected regret.}:

\begin{clm}[No-Regret Implies CCE]\label{clm:cce}
	Suppose that players each use a NRL strategy that guarantees average regret $\epsilon(T)$. Fix a horizon $T$ and let $\boldsymbol{\alpha}^1,...,\boldsymbol{\alpha}^T$ be the probability vectors for the \emph{joint} distribution of action profiles induced by each players' play. Then
	the following compound distribution, which we call the \emph{average empirical action distribution}, is an $\epsilon(T)$-approximate CCE:

\begin{align*}
	\bar{\mathcal{D}} \coloneqq \boldsymbol{\alpha}^{\tilde{t}}, \ \ \tilde{t} \sim \text{Uniformly}\{1,2,...,T\}
	\end{align*}
\end{clm}
Note that if $\mathbf{a} \sim \bar{\mathcal{D}}$, then $\mathbf{a}$ is drawn from the joint distribution over actions profiles at a uniformly selected time period. The proof this claim is well-known and almost immediate; see e.g. \cite{roughgarden2016twenty}. 

We can thus use Claim \ref{clm:cce} to obtain a measure of the performance of our mechanisms under a CCE. Our high-level approach will be to simulate a repeated game corresponding to that induced by each of the mechanisms we study, allow agents to learn to bid using NLR algorithms, and estimate each mechanism's equilibrium welfare, revenue, and empirical price of anarchy by sampling from the average empirical joint distribution. The particular NLR algorithm we choose is the Exponential Weights (EW), one of the most well-studied algorithms with well-developed guarantees. Algorithm \ref{alg:ew} provides pseudocode for the basic (EW) algorithm applied to our setting:

\begin{algorithm}[H]
	\caption{Exponential Weights
		\label{alg:ew}}
		\KwIn{Learning Rate $\eta$, Bid Space $\mathcal{B} = \{0,\frac{1}{d},\frac{2}{d},...1\}$, Number of Steps $T$, Mechanism $M$}
		
	 $\mathcal{W}(b)^0 \gets \frac{1}{d+1}$ for each $b$ in $\mathcal{B}$.\\
	\For{$t \in 1...T$}{
		 Draw $b^{t} \sim \mathcal{W}^t$.\\
		 Submit $b^t$ to $M$. \\Experience utility $u_i(b^t)$ from mechanism.\\
		 Query $M$ to obtain counterfactual utility $u(b')^t$ for all alternative bids $b'$ in $\mathcal{B}$. \\
		 Update weights using $\mathcal{W}^{t+1}(b) = \exp(\eta u(b)^t) \cdot\mathcal{W}^{t}(b)$ for all $b$\\
		 Renormalize weights.
	}
\end{algorithm}

The choice of the learning rate parameter $\eta$ will affect the performance of the algorithm. A basic analysis of EW is available in many books and lecture notes (e.g. \cite{roughgarden2016twenty}), so we state the following claim without proof:

\begin{proposition}[EW Optimal $\eta$] \label{clm:eweta} Under the EW algorithm, the optimal choice of $\eta$ for a known-horizon setting with $K$ actions is $\eta= \sqrt{\frac{\ln K}{T}}$, resulting in worst-case expected regret $2 \sqrt{\frac{\ln K}{T}}$. Thus to guarantee $\epsilon$ expected regret (and so achieve an $\epsilon$-CCE), we need $T=\frac{4\ln K}{\epsilon^2}$ rounds of learning. 
\end{proposition}

\paragraph{Protocol for Experiment 2 and 3.} Algorithm \ref{exp:twothree} gives pseudocode for the protocol as a whole. At a high level, the component steps are simply: sampling a valuation for the bidders; running EW for the desired number of learning rounds; after completing learning, re-sample a strategy profile from the bidder strategy distributions from a randomly sampled round and measure the revenue, welfare, and price of anarchy. We run $N_l$ rounds of learning, re-sample $N_t$ times, and repeat this for $N_s$ valuation samples. Using the guarantees of Claim \ref{clm:eweta}, we can decide how stringent we want our CCE to be (i.e. how much ``approximation'' we allow in the ``approximate'' CCE) and then set $N_l$ and $\eta$ accordingly; we then select $N_t$ based on how exact we would like an estimate. We highlight, though, that there is a significant cost to increasing the various parameters of the experiment. For instance, each additional learning step actually requires $d*M+1$ auctions (i.e. the original and all counterfactual bids for all bidders); put together with the $N_t$ samples from the time-averaged distribution,  this implies that each additional valuation sample requires running $N_l(dM+1)+N_t$ full auctions for each mechanism. We are thus constrained in the length of experiments we could run. We provide detailed descriptions and chosen values for these and other parameters in Section \ref{s:app-expdetails}.

\begin{algorithm}
	\caption{Protocol for Experiments 2 and 3
		\label{exp:twothree}}
	\SetKwInOut{Input}{input}\SetKwInOut{Output}{output}
	\Input{Bid dataset $\mathscr{D}$; Bid discretization number $d$; Number of learning steps $N_{l}$, Number of Test steps $N_{t}$, Number of sampling steps $N_s$, Number of slots $S$, Bidder set $\mathcal{I}$, Number of Bidders $M$, EW Learning Rate $\eta$}
		
		\For{$(\A, \P_{\A}) \in \{\text{[Greedy,Opt]} \times \text{[GSP,VCG]} \}$}{
		\For{$t \in 1,2...,N_s$}{
		 Draw $M$ valuations randomly from $\mathscr{D}$.\\
		  Initialize bidders with values and fresh Exp. Weights($\eta$). Store each in $\mathcal{W}_{is}^t$.\\		 
		\For{$t \in 1,2,...,N_l$}{
		 Initialize and run an $S$-slot auction using $(\A, \P_{\A})$\\
		 Draw bids and fix them.\\
		\For{$i \in \mathcal{I}$}{
		\For{$b' \in \{\frac{1}{d}, \frac{2}{d}, ...,1\}$}{
		 Re-run auction with all other players' bids fixed, but player $i$ using $b'$\\
		 Save payoff \\
		}
	 Update ExpWeights. \\
		}
		}

		\For{$t' \in 1,2,...N_t$}{
		 Draw round number $r_t$ uniformly at random from $\{1,...N_L\}$ \\
		 Draw bid for each bidder from bidder distributions at that round\\
		 Run auction with these bids\\
		 Calculate and save round info   \\
		}
	Store sample results\\
	}	
}
\end{algorithm}

\paragraph{Random Valuation Setting.} To study an empirical version of the Bayes-Nash equilibria we analyzed, we use the \emph{population interpretation} of Bayesian games. In this interpretation, we view each valuation as a separate player from a population (which corresponds to the original player); so for instance, rather than having two players valuations distributed uniformly, we would have two populations of player with valuations distributed with equal probability over the discrete uniform distribution. The game then consists of nature first selecting a player from each population and then each player acting as a standard agent (i.e. playing a normal form game rather than a Bayesian game). This interpretation has a long history, but \cite{hartline2015no} show that the utility of NRL algorithms extends to this interpretation. In particular, they show that that if each individual type uses a NRL algorithm, then any convergent subsequence of the sequence of strategy distributions converges almost surely to a (Bayesian) CCE. They also prove that the social welfare of play enjoys the welfare guarantees of stage game if it is smooth via an extension argument. 

We use a similar approach to learn a BCCE in our setting. There are multiple definitions of BCCEs \cite{forges1993five}, but the one used by \cite{hartline2015no} is a natural one and corresponds to a standard CCE over the population of the game. (It is thus quite intuitive that NRL would produce a correlated equilibrium in this setting.) One addition we make is that of a \emph{uniform exploration period}. This period is inspired by \cite{feng2021}, which shows that adding a period of uniformly random exploration (and updating weights based on the exploration) before following the recommendations of the NRL algorithm results in provable convergence to specific, natural equilibria\footnote{For instance, in a first price auction, they recover the classic symmetric profile of each bidder bidding a $1/(N)$ fraction of the value, where $N$ is the number of bidders. Convergence to \emph{specific} equilibria is not generally guaranteed or easy to prove, so this exploration period does add value. }.

\paragraph{Protocol for Experiment 1.} We imagine the two-bidder, two-auction uniform distribution case from Section \ref{s:eq}, but now view the bidders as \emph{populations}. We discretize the uniform distribution over $[0,1]$ into $[0, \frac{1}{V},\frac{2}{V},...1]$; every \emph{population representative} for both populations has one of those valuations. When a round occurs, we randomly select a population representative from each population, and each representative has an equal probability of selection. (Hence, as the discretization becomes increasingly fine, we approach the continuous uniform distribution.) 

If the round $t$ is less than the number of exploration steps $N_e$, selected population representatives choose a random action; otherwise, they sample a bid from their EW distribution. In either case the representatives playing update their EW distributions. 

We run this for a large number of rounds, each round recording the valuations of the selected population representatives and how they bid. We then compute the average bid for each population representative over the non-exploration period and use this as their average bid.

\begin{algorithm}
	\caption{Protocol for Experiment 1
		\label{exp:bayes}}
	\SetKwInOut{Input}{input}\SetKwInOut{Output}{output}
	\Input{Value discretization $V$, Bid discretization number $d$; Number of learning steps $N_{l}$,  Number of slots $S$, Bidder Populations $\mathcal{I}$, Number of Bidders $M$, EW Learning Rate $\eta$}
	
	\For{$(\A, \P_{\A}) \in \{\text{[Greedy,Opt]} \times \text{[GSP,VCG]} \}$}{
		Initialize set of population representatives  $\mathcal{V}_i:= \{0, \frac{1}{V}, \frac{2}{V}...,1\}$ for each player $i$ \\
		For each valuation, initialize a separate instance of EW with weight $\eta$ \\
			\For{$t \in 1,2,...,N_l$}{
				For each population $i$, draw index $j$ uniformly from  $\{0,1...,V\}$ \\
				
				Initialize and run an $S$-slot auction using $(\A, \P_{\A})$
				with bidder with valuation $\frac{1}{j_i}$ from each population $i$
				 \\
				Draw bids from each bidder's EW distributions as of time $t$ \\
				\For{each drawn bidder}{
					\For{$b' \in \{\frac{1}{d}, \frac{2}{d}, ...,1\}$}{
						Re-run auction with other bidders' bids fixed, but player $i$ using $b'$\\
						Save payoff \\
					}
					Update ExpWeights for bidder\\
				}
			}

	}
		Return dataset of drawn bidder and bids
\end{algorithm}

\subsection{Experiment Details}\label{s:app-expdetails}

Table \ref{table:exp} provides more details on the experimental parameters we chose. Since each sample runs many auctions, each of which is expensive to run, the parameter choices must balance the experiment complexity (i.e. how many bidders, slots, how many bid choices) with the the computational constraints. Before turning to random value setting, we briefly describe each parameter:
\begin{itemize}
\item {\textbf {Equilibrium Concept.}} Whether we are searcing for CCE or BCCE.
\item {\textbf {Bid Data.}} Whether we draw data form the Random Advertisers vs. Random Auction datasets, or simply use the discretized uniform distribution (Synthetic). 
\item {\textbf {Bid Discretization.}} The number (excluding $0$) of evenly-spaced bids into which we divide the bidspace. 
\item{\textbf {Value-dependent Discretization.}} Whether bid options are evenly spaced between $0$ and a bidder's valuation rather than just evenly spaced between $0$ and $1$. 
\item{\textbf {Number of Bidders.}} How many bidders in the experiment. For the case of the Bayesian Setting (Experiment 1), we count each valuation separately. 
\item{\textbf {Number of Slots.}} How many slots are available.
\item{\textbf {Number of Learning Steps.}} How many rounds are used for learning. (In the Experiment 1, this is total, not per bidder-valuation pair.)
\item{\textbf {Number of Exploration Steps.}} How many rounds pre-EW are used for pure random exploration. 
\item{\textbf {Number of Test Steps.}} After the learning phase, how many observations from the average empirical distribution we use to estimate relevant quantities.
\item{\textbf {Allow Overbidding.}} Whether we allow bidders to bid more than their value. (If we do, bid choices are evenly spaced between $0$ and $2$ if bids are not value-dependent and $0$ and twice the valuation if they are.) 
\item{\textbf {Geometric Discount Constant.}} The discount constant, i.e. discount for the first slot for each bidder.
\item{\textbf {Geometric Discount Constants.}} Constant factors in the bidders' discount curves.
\end{itemize}
\footnotesize
\begin{table}[h!]
	\centering
	\begin{tabular}{||c c c c c||} 
		\hline
		Variable & Description & Experiment 1& Experiment 2 & Experiment 3 \\ [0.5ex] 
		\hline\hline
		Eq? &Equilibrium Concept &  BCCE & CCE & CCE \\
		$\mathcal{D}$ & Bid Datasets & Synthetic & Random Advertisers & Random Auction \\	
		$d$ &Bid Discretization & 20 & 20 & 20 \\ 
		$\mathscr{D}(v)$? & Discretization depends on value? & \greencheck & \greencheck & \greencheck \\
		$M$& Number of Bidders &2 &9 &9 \\
		$|\{v_b\}|$ & Number of valuations per bidder & 11 & 1 &1 \\
		$S$ &Number of Slots & 2& 4&4 \\
		$N_s$ &Number of Sampled valuations & NA& 200 &200\\
		$N_l$ &Number of Learning Steps & 500000 &  100&100  \\
		$N_t$& Number of Testing Steps& NA &  200& 200  \\
		$N_e$&Number Exploration Steps &10367    &0  &0  \\
		
		OB? &Allow Overbidding? &  \greencheck & \redx & \redx  \\
		$\mathbf{\delta}_0$ & Geometric Discount Constant & 1 & 1 & 1 \\ 
		$\mathbf{\delta}$ & Geometric Discount Factors  & \tiny$[.37,.85]$ & \tiny $[0.9, 0.9,0.8,0.8,0.7,0.7,0.6,0.6,0.5]$ &  \tiny  $\leftarrow \leftarrow \leftarrow $\\ 		[1ex] 
		\hline
	\end{tabular}
	\caption{Experimental Parameters}
	\label{table:exp}
\end{table}
\normalsize
\section{Detailed Proofs}

\subsection{Proofs from Section \ref{s:poa}}\label{s:app-smoothproofs}

\begin{proof}[Proof of \Cref{lem:partialmon}]
	First consider player $i$. Since $i$ increased his bid between $b$ to $b'$,  he achieves some slot $\sigma'$ at least as high as $\sigma$. 
	
	Now, consider slots above $\sigma'$. By definition, $i$ has not placed an effective bid higher than bidders occupying those slots (or else he would have been placed in that slot or above). So $i$'s deviation leaves  unchanged the bidder allocation and so valuations for those slots. Now, at $\sigma'$, by construction, we must have that
	\begin{align*}
		\delta_{\tau(i)}^{\sigma'} b_{i}' \geq \delta_{\tau(\pi(\sigma', b))}^{\sigma'} b_{\pi(\sigma', b)}
	\end{align*}
	or else $i$ would not have been assigned to $\sigma'$. So the desired inequality holds for this slot. 
	
	Finally, consider each slot $s'$ between $\sigma'$ and $\sigma$. Notice that the set of bidders unallocated when $s'$ is considered under $b'$ has only changed by \emph{losing} $i$ and possibly \emph{gaining} either $\pi(\sigma', b)$ or a displaced previous winners from slots between $\sigma'$ and $\sigma$ due to $\pi(\sigma', b)$ being displaced by $i$ and any cascading effects. But this means that in particular $\pi(s',b)$ remains unallocated when $s'$ is considered. Hence, if $\pi(s',b') \neq \pi(s',b)$, it can only be because the assigned bidder under $b'$ had higher discounted value than the bidder assigned there under $b$. Since this holds for any $s'$ in the range, the claim holds. 
	
\end{proof}

\begin{proof}[Proof of \Cref{thm:lbgreedyvcg}]
	
	Let $v_A=1+\epsilon$, $v_B = 1$, $v_C = 1-\epsilon$. Let ${\boldsymbol {\delta}_A} = (1, 1,1-2\epsilon)$, ${\boldsymbol{\delta}}_B=(1,1,0)$, $\boldsymbol{\Delta}_C=1, \epsilon, \epsilon^2$. 	
	The welfare of $(C,B,A)$ is $3-2\epsilon -2\epsilon^2$, while the welfare of $(A,B,C)$ is $2+\epsilon+\epsilon^2 -\epsilon^3$.
	
	Suppose that each player bids their value, ie:
	\begin{align*}
		b^* = (b_A,b_B,b_C) = (v_A,v_B,v_C) = (1+\epsilon, 1, 1-\epsilon).
	\end{align*}
	We claim this is an equilibrium and results in $(A,B,C)$. The allocation follows since the allocation algorithm is greedy in bids. To see that this is an equilibrium, first consider what values each player is getting: A gets $1+\epsilon $, B gets $1$, C gets $(1-\epsilon) \epsilon^2  = \epsilon^2 - \epsilon^3$. With these, we can calculate what prices each player is paying: Player C pays nothing, since he is imposing no externality on A or B. B is imposing an externality on C - without B, C would get the second slot for a valuation of $\epsilon-\epsilon^2$  and B imposes no externality on A. So B will be charged $\epsilon-\epsilon^2$. Finally, A imposes the same externality on C (because without A, B would get the first slot, so C would get the second slot) and imposes no externality on $B$. 
	
	So the payoffs are:
	\begin{align*}
		\pi_A(b^*) &= 1+\epsilon - \epsilon +\epsilon^2 = 1+\epsilon^2 \\
		\pi_B(b^*) &= 1 -\epsilon+\epsilon^2\\
		\pi_C(b^*) &= \epsilon^2
	\end{align*}
	Notice that these are always positive. (The only one that could possibly be negative would be $\pi_B$, but if $\epsilon <1$, then $1-\epsilon >0 \implies \pi_B>0$; if $\epsilon>1$, then $\epsilon^2-\epsilon >0 \implies \pi_B>0$.)
	
	Now we consider possible deviations. Start with A. While there are an uncountable number of deviations in bid space, they are all equivalent but for their effects on A's position and price. So notice that if $A$ were to move to second position by bidding $b_A'$ less than $b_B$ but more than $b_C$,  it would receive the same payoff, because its discount rate is $1$ and it imposes the same externality as before, so no such bid could improve A's payoff. If A were to bid $b_A'$ less than $b_C$, it could get the third slot at a price of 0, but it would only get $1-2\epsilon < 1+\epsilon^2 = \pi_A(b^*)$. So A has no profitable deviations. For B, improving his position cannot improve his payoff or change his externality, and moving to slot 3 would result in $0$ payoff, while he currently makes positive profit. For Player C, notice that first of all, if we rule out overbidding, Player C cannot improve his position; but suppose we do not rule this out. By moving to Slot 2 (by bidding, say, $b_C= 1+\epsilon/2$) C would exert an externality of $1$ on Player $B$ and so get negative payoff ($1-\epsilon -1 = -\epsilon$). By moving to Slot 1 (by bidding $b_C \geq 1+\epsilon$) C would exert the same externality on $B$ and so again receive negative payoff.
	
	Hence, $b^*$ is an equilbrium. But then we have that:
	\begin{align*}
		\frac{EQ}{OPT} = \frac{2+\epsilon+\epsilon^2 -\epsilon^3}{3-2\epsilon - 2\epsilon^2}
	\end{align*}
	which comes arbitrarily close to $2/3$ for small enough $\epsilon$.
\end{proof}


We now turn to proving \Cref{thm:optgspexample}. We will proceed in several steps:first we characterize what must hold in equilibrium. Then we provide examples that meet this. Finally, we optimize this bound. We will break this up into several propositions before the main proof.

\begin{proposition} \label{prop:optgsp-poa-eq}
	Let A have discount curve $(1,\delta_A)$, and B have discount curve $(1,\delta_B)$, with $\delta_A< \delta_B$ (so $\Delta : = \frac{1-\delta_B}{1-\delta_A} <1$). Now suppose that $\Delta^2 v_B \leq v_A \leq \Delta v_B \leq \frac{v_A}{\Delta} \leq v_B$\footnote{As is always nice to check, we are not reasoning about an empty set. Consider $v_B=1$, $v_A=\frac{1}{2}$, $\delta_A = \frac{1}{2}$, $\delta_B= \frac{2}{3}$.}.  Then the following strategy profile is an equilibrium:
	\begin{align*}
		\mathbf{b}^* = \left(\Delta(1-\delta_B)v_B + \epsilon, \Delta(1-\delta_B) v_B \right)
	\end{align*}
	for any $\epsilon>0$, and for small enough $\epsilon$ neither bidder is overbidding. The auctioneer then selects $(A,B)$, but $(B,A)$ would be optimal. 
\end{proposition}

\longversion{
	Before we prove that this claim, we first show that we are not reasoning about an empty set. Consider $v_B=1$, $v_A=\frac{1}{2}$, $\Delta = \frac{2}{3}$ (which, for example, can be obtained by $\delta_A = \frac{1}{2} < \frac{2}{3} = \delta_B$). Then $\Delta^2 v_B$ = $\frac{4}{9} < \frac{1}{2}=v_A$, so the first inequality holds. $v_A = \frac{1}{2} \leq \frac{2}{3} = \Delta v_B$, so the second inequality holds. $\Delta v_B = \frac{2}{3} \leq \frac{3}{4}=\frac{1/2}{2/3}= \frac{v_A}{\Delta}$, so the third inequality holds, and $\frac{v_A}{\Delta} =\frac{3}{4} < 1=v_B$ so the final inequality holds. 
	
	Notice that under this particular example, if the auctioneer selects $(A,B)$ as claimed (and the bids truly form an equilibrium), we get a competitive ratio of:
	\begin{align*}
		\frac{EQ}{OPT} = \frac{v_A + \delta_B v_B }{v_B + \delta_A v_A} = \frac{1/2 + 2/3}{1+1/2*1/2} = \frac{7/6}{5/4} = \frac{28}{30}.
	\end{align*}
	So we will proceed to prove the claim, and then optimize the ratio. 
}

\begin{proof}[Proof of  \Cref{prop:optgsp-poa-eq}]
	The auctioneer selects $(A,B)$ whenever 
	\begin{align*}
		b_A + \delta_B b_B \geq b_B + \delta_A b_A \iff b_A \geq \Delta b_B.
	\end{align*} But
	\begin{align*}
		b_a =  \Delta (1-\delta_B) v_B +  \epsilon \geq  \Delta^2 (1-\delta_B) v_B =\Delta b_B
	\end{align*}
	where the inequality follows from the fact that $\delta_A < \delta_B \implies \Delta <1$. So the outcome is that A gets the top slot; since A will win as long as $b_A \geq \Delta b_B$, A will be charged $\Delta b_B$. B will receive the second slot, and be charged nothing. On the other hand,  we note that (B,A) is optimal iff:
	\begin{align*}
		v_A + \delta_B v_B \leq v_B + \delta_A v_A  \iff v_A \geq \Delta v_B.
	\end{align*}
	This holds by assumption, so (B,A) is in fact the optimal allocation. 
	
	Now we consider possible deviations from the bid profile. For $A$, bidding higher does not change the allocation nor the payment, and bidding lower than its bid but more than $b_B$ also does not affect the allocation or the payment, so the only deviation to consider is bidding less than $b_B$. If it does this, it will change the allocation to $(B,A)$ and get $\delta_A v_A$ while paying nothing, but:
	\begin{align*}
		v_A - \Delta b_B& = v_A - \Delta^2(1-\delta_B) v_B \geq v_A - v_A (1-\delta_B) \\&= v_A (1- (1-\delta_B)) = \delta_B v_A > \delta_A v_A
	\end{align*} 
	where the first equality follows by the pricing rule and strategy profile, the first inequality follows from the fact that $v_A \geq \delta^2 v_B \implies -\Delta^2 v_B \geq -v_A$), and the final inequality by assumption. So deviating to be assigned  the second slot would not be profitable for A. 
	
	Now consider B. Again, the only deviations that we must consider are those which change the allocation to $(B,A)$. But if B were to deviate to such a bid, he would be charged $b_A/ \Delta$. But we have that:
	
	\begin{align*}
		\frac{b_A}{\Delta} &= \frac{\Delta(1-\delta_B) v_B +\epsilon}{\Delta} = (1-\delta_B) v_B + \frac{\epsilon}{\Delta} \\&\implies v_B - \frac{b_A}{\Delta} = v_B - (1-\delta_B) v_B - \frac{\epsilon}{\Delta} = \delta_B v_B -\frac{\epsilon}{\Delta} < \delta_B v_B
	\end{align*}
	so this deviation would not be profitable for B. 
	
	Now, note that B is trivially not overbidding since $\Delta, 1-\delta_B <1$. To show that there exists a small enough $\epsilon$ so that $A$ is not overbidding, note that we need:
	\begin{align*}
		v_A - \Delta(1-\delta_B) v_B - \epsilon \geq 0
	\end{align*}
	so it is enough that $v_A -\Delta(1-\delta_B) >0$. But:
	\begin{align*}
		\Delta = \frac{1-\delta_B}{1-\delta_A} > 1- \delta_B &\implies -\Delta \leq -(1-\delta_B) \\&\implies - \Delta^2 < - \Delta(1-\delta_B)\\&\implies -v_B \Delta^2 < -v_B \Delta(1-\delta_B)
	\end{align*}
	But then 
	\begin{align*}
		v_A - v_B \Delta (1-\delta_B) > v_A - \Delta^2 v_B \geq 0
	\end{align*}
	as desired, where the last inequality follows by assumption. 
	Thus, we have shown that this bid profile is an equilibrium that achieves suboptimal welfare.
\end{proof}

Now we turn to optimizing this bound. 

\begin{proposition}\label{prop:optimizing}
	There exists a choice $v_A, v_B$, $\delta_A < \delta_B$, such that the bid profile above is an equilibrium and obtains welfare arbitrarily close\footnote{We leave it here to avoid tie-breaking issues.} to 3/4 of the optimal welfare. This implies that the Price of Anarchy is at least $4/3$. 
\end{proposition}
{
	\begin{proof}
		Let $v_B = 1$, and let $\delta_B=\frac{1}{2}$. We won't fix $\delta_A$, but rather we will assume that $\delta_A< \delta_B =\frac{1}{2}$ let it approach $0$. We also set $v_A$ as a function of $\delta_A$: $v_A = \frac{1}{4(1-\delta_A)^2}$.
		
		Now notice that $r = \frac{v_A}{v_B} = \frac{1}{4(1-\delta_A)^2} = \frac{(1/2)^2}{(1-\delta_A)^2} = \frac{(1-\delta_B)^2}{(1-\delta_A)^2}=\Delta^2$. Then
		\begin{align*}
			\Delta^2 \leq r \leq \Delta \leq r/\Delta \leq 1\text{ and } \Delta^2 v_B \leq v_A \leq \Delta v_B \leq \frac{v_A}{\Delta} \leq v_B.
		\end{align*}
		Hence, the hypotheses of  \Cref{prop:optgsp-poa-eq} are satisfied, so the equilibrium described is an equilibrium. Then the competitive ratio is given by:
		\begin{align*}
			\frac{EQ}{OPT} = \frac{\frac{1}{2} + \frac{1}{4(1-\delta_A)^2}}{1 + \frac{\delta_A}{4(1-\delta_A)^2}} = \frac {2(1-\delta_A)^2 + 1}{\delta_A + 4(1-\delta_A)^2}
		\end{align*}
		Notice that at $\delta_A =0$, this quantity is $\frac{3}{4}$, and is $1$ at $\frac{1}{2}$. But notice also that the denominator, viewed independently, is a quadratic function with only complex roots. Thus, the fraction is continuous. Since it varies continuously from $1$ to $3/4$, it must pass through every point arbitrarily close to $3/4$ from the right. 
		Hence, we can achieve competitive ratio arbitrarily close 3/4, so the Price of Anarchy is at least 4/3. 
		
	\end{proof}

\begin{proof}[Proof of \Cref{thm:optgspexample}]
	Combining \Cref{prop:optgsp-poa-eq} and \Cref{prop:optimizing} yields the claim.
\end{proof}

\subsection{Proofs from Section \ref{s:eq}}

\begin{proof} [Proof of \Cref{thm:eqoptgsp}]
	We follow the same structure as the proof of Theorem \ref{thm:eqggsp}. Consider the problem from A's perspective, and suppose that B is using a linear strategy $\beta v_B$. (In the theorem statement, $\beta = 1-\delta_B$, but as in Theorem \ref{thm:eqggsp}, A's best-response will not depend on $\beta$ being $(1-\delta_B)$, so we leave it free.) Now, note that the winning condition is that:
	\begin{align*}
		(A,B) \succeq (B,A) \iff b_A \geq \Delta b_B = \Delta \beta v_B
	\end{align*}

	Suppose that $A$ wins the top slot with a bid $b_A$. For now, suppose that $b_A$ is less than $\beta \Delta$. Then the expected payment is:
	\begin{align*}
		\E[\Delta b_B| \Delta b_B \leq b_A] &= \E[ \Delta \beta v_B| v_B \leq \frac{b_A}{\Delta \beta}] \\
		&= \beta \Delta \E[v_B|v_B \leq \frac{b_A}{\Delta \beta}] = \beta \Delta \frac{b_A}{2 \Delta \beta} = \frac{b_A}{2}
	\end{align*}
	where the second inequality follows from the properties of the uniform distribution. 
	Then by choosing any $b_A$, A gets the expected payoff:
	\begin{align*}
		\E[u_A(b_A|v_A)] &=\left (v_A- \frac{b_A}{2}\right)\Pr[b_B \leq b_A] + \delta_A v_A (1-\Pr[b_B \leq b_A])
		\\ &= \left(v_A- \frac{b_A}{2}\right) \frac{b_A}{\beta \Delta } + \delta_A v_A\left(1- \frac{b_A}{\beta\Delta}\right) \\&= \frac{v_A b_A}{\beta \Delta} - \frac{b_A^2}{2 \beta \Delta} + \delta_A v_A - \frac{b_A\delta_A v_A}{\beta \Delta}
	\end{align*}
	
	Taking the derivative, the first order conditions requires that
	\begin{align*}
		\frac{v_A}{\beta\Delta} - \frac{b_A}{\beta \Delta} - \frac{b_A \delta_A v_A}{\beta \Delta} = 0 \iff b_A = (1-\delta_A)v_A.
	\end{align*}
 As the second derivative is negative, this is a maximum. 
	
	On the other hand, if $b_A > \beta \Delta$, $A$ wins with probability $1$ and pays $\E[\Delta \beta v_B] = \frac{\beta \Delta }{2}$, getting total payoff $v_A - \frac{\beta \Delta}{2}$. Again, notice that this is the same as the value taken on by the other expression above if $b_A = \beta \Delta$, and bidding any $b_A > \beta \Delta$ results in the same payoff as bidding $\beta \Delta$. So as before, $A$ need only consider maximizing his utility over $b_A \in [0, \beta \Delta]$; the maximum can be either at $0$, $\beta \Delta$, or the critical point (which is a maximum), or the endpoints. But the payoff is increasing for all point left of the critical point and decreasing for all points right of it; hence, as before, if the critical point is left of $\beta \Delta$, it is an interior maximum, and if it after $\beta \Delta$, then it is just as good as bidding $\beta \Delta$. 
	
	Hence, $b_A=(1-\delta_A) v_A$ is a best response, and the linear portion is unique whenever $b_A = (1-\delta_A) v_A \leq \beta \Delta \implies v_A \leq \frac{\beta \Delta}{1-\delta_A}$.  That is, the linear coefficient of $(1-\delta_A)$ is unique for all $v_A \in [0,  \frac{\beta \Delta}{1-\delta_A}]$, but any bid of at least $\beta \Delta$ is a best-response for $v_A > \frac{\beta \Delta }{1-\delta_A}$. And notice that nothing about this depended on $\beta$; but if $\beta = (1-\delta_B)$, then the kink in bidding occurs at $v_A=\Delta^2$. 

	Again, viewing this from B's perspective will give the same set of computations, mutatis mutandum, so we conclude that $\left((1-\delta_A)v_A, (1-\delta_B)v_B\right)$ is an equilibrium. A similar uniqueness argument holds for B's strategy as well. Thus, this equilibrium is the unique linear equilibrium.  
\end{proof}

\begin{proof}[Proof of \Cref{prop:optgsprev}]
	Note that A wins iff $b_A > \Delta b_B $. Since $b_B = (1-\delta_B) v_B$ and $b_A = (1-\delta_A) v_A$, A wins whenever
	\begin{align*}
		v_A \geq v_B\Delta^2 .
	\end{align*}
	Now, if A wins, she pays the minimum price $p$ such that $p \geq \Delta b_B$, which is just $\Delta b_B = v_B \frac{(1-\delta_B)^2}{1-\delta_A}$. Similarly, if B wins, he pays the minimum price $p$ such that $p \geq \frac{b_A}{\Delta}= \frac{(1-\delta_A)^2}{1-\delta_B} v_A$. So, we can write $R(v_A,v_B)$ as:

	\begin{align*}
		R(v_A,v_B) 
		= \begin{cases} (1-\delta_B)\Delta v_B & v_A \geq  \Delta^2 v_B \\ \frac{(1-\delta_A)}{ \Delta}  v_A & v_A \leq  \Delta^2 v_B 
		\end{cases}
	\end{align*}
	Now we can calculate the expected revenue by again writing it as a piecewise integral:
	\begin{align}
		\begin{split}\label{eq:revoptgs}
		\E[R(v_A,v_B)] &= 
		\int_0^1 \int_0^{\Delta^2 v_B} \frac{1-\delta_A}{\Delta} v_A dv_A dv_B + \int_0^1 \int_{\Delta^2 v_B} ^1 (1-\delta_B) \Delta v_B dv_A dv_B
		\end{split}\\ \nonumber
		&= \frac{1-\delta_A}{\Delta} \int_0^1 \frac{v_A^2}{2} \biggr|_{0}^{\Delta^2v_B} dv_B + (1-\delta_B)\Delta \int_0^1  v_B v_A \biggr|_{\Delta^2 v_B}^1 dv_B \\ \nonumber
		& = \frac{1-\delta_A}{\Delta} \int_0^1 \frac{\Delta^4 v_B^2}{2} dv_B + (1-\delta_B) \Delta \int_0^1 v_B - \Delta^2 v_B^2 dv_B\\ \nonumber
		& = \frac{1-\delta_A}{\Delta} \frac{\Delta^4}{2} \frac{v_B^3}{3} \biggr|_{0}^{1} + (1-\delta_B) \Delta \left[\frac{v_B^2}{2} -\frac{\Delta^2 v_B^3}{3}\right]\biggr|_0^1\\ \nonumber
		& = \frac{\Delta^3(1-\delta_A)}{6} + (1-\delta_B) \Delta \left[\frac{1}{2} - \frac{\Delta^2}{3}\right] \\ \nonumber
		&= \frac{\Delta^3(1-\delta_A)}{6} + \frac{\Delta(1-\delta_B) }{6}\left(3 - 2 \Delta^2\right)
	\end{align}
	as claimed.
\end{proof}

\begin{proof}[Proof of Theorem \ref{thm:eqgreedygsp}]
	Suppose B bids with $b_B = \Delta v_B$.  Then since $b_B \leq \Delta$, any bid A makes above $\Delta$ will be equivalent in that she will certainly win and pay the same price. Thus we can write A's win probability and expected payment given winning as:
	\begin{align*}
		\Pr[b_B \leq b_A] = \begin{cases}
			\frac{b_A}{\Delta} & b_A \leq \Delta \\
			1 & b_A > \Delta
		\end{cases} \text{ and } \E[v_B| b_B \leq b_A] = \begin{cases}
			\frac{b_A}{2\Delta} & b_A \leq \Delta \\ 1 & b_A > \Delta
		\end{cases}
	\end{align*}
	Hence A's payoff is:
	\begin{align*}
		u_i(b_A) &= \begin{cases}
			\left(v_A - (1-\delta_B)  \Delta\frac{b_A}{2\Delta}\right) \frac{b_A}{\Delta} + \delta_A v_A(1-\frac{b_A}{\Delta}) & b_A \leq \Delta \\
			v_A - \frac{(1-\delta_B) \Delta}{2}\end{cases}\\
		&= \begin{cases}
			\left(v_A - \frac{1-\delta_B}{2} b_A\right) \frac{b_A}\Delta + \delta_A v_A (1-\frac{b_A}{\Delta}) & b_A \leq \Delta \\
			v_A - \frac{\Delta(1-\delta_B)}{2} & b_A > \Delta 
		\end{cases}
	\end{align*}
	Notice that at $b_A = \Delta$, these values coincide; beyond $\Delta$, any value that A bids results in the same payoff. So, this payoff function is a sort of capped quadratic in $b_A$ with the kink at $\Delta$. Thus, to find the optimal bid, A need only compare any inner critical point with the end point (which it would even in the absence of such a kink given it were maximizing over a closed set). 
	
	On the interior section, A's first order condition is:
	\begin{align*}
		\frac{v_A}{\Delta} - \frac{(1-\delta_B)b_A}{\Delta} -\frac{\delta_A v_A}{ \Delta} = 0& \implies
		(1-\delta_B) b_A = v_A(1-\delta_A) \\ &\implies b_A = v_A  \frac{1-\delta_A}{1-\delta_B} = \frac{v_A}{\Delta}.
	\end{align*}
	As usual, concavity gives that this is a local maximum. 
	
	But now notice that $u_i(b_A)$ is continuous up until $b_A = \Delta$, where it coincides with the next piece. Moreover, it is concave (strictly, on $[0, \Delta]$); hence, if a local maximum is reached, it must be a maximum over the interval $[0,\Delta]$, \emph{including} the point at $\Delta$. 
	
	So, whenever $\frac{v_A}{\Delta} \leq \Delta \iff v_A \leq \Delta^2$, it is immediate that A can do no better than bidding $b_A = v_A/\Delta$. On the other hand, if $v_A \geq \Delta^2$, then $\frac{v_A}{\Delta} \geq \Delta$. But above $\Delta$, increasing the bid does not improve A's payoff, and so the choice of $v_A/\Delta$ prescribes a bid higher than necessary - bidding $\Delta$ would suffice. However, it also does not hurt A's payoff.
	
	Thus, bidding $b_A = \frac{v_A}{\Delta}$ is always a best-response to B bidding $b_B = \Delta v_B$ (though it is not a unique best-response). 
	
	Now we do a similar calculation from B's perspective, supposing that $b_A = \frac{v_A}{\Delta}$. B wins if $b_A \leq b_B$ and pays $(1-\delta_A) b_A$. Again, we shall consider for the possibility of overbidding, and write the win probability and expected payoff that B will receive for any bid as:
	\begin{align*}
		\Pr[b_A \leq b_B]=\begin{cases} \Delta b_B & b_B \leq \frac{1}{\Delta} \\1 & b_B \geq \frac{1}{\Delta} \end{cases} \text{ and } \E[v_A |b_A \leq b_B] = \begin{cases} \frac{b_B \Delta}{2} & b_B \leq \frac{1}{\Delta} \\ \frac{1}{2} & b_B \geq \frac{1}{\Delta}\end{cases}
	\end{align*}
	
	Then we have that
	\begin{align*}
		u_B(b_B) &= \begin{cases}
			\left(v_B - (1-\delta_A) \frac{1}{\Delta} \frac{b_B \Delta}{2}\right) \Delta b_B + \delta_B v_B(1-\Delta b_B) & b_B \leq \frac{1}{\Delta}\\
			v_B - \frac{1-\delta_A}{2\Delta} 
		\end{cases}\\&= \begin{cases}
			\left(v_B - (1-\delta_A)\frac{b_B}{2}\right) \Delta b_B + \delta_B v_B - \Delta \delta_B v_B b_B& b_B \leq \frac{1}{\Delta}\\
			v_B - \frac{1-\delta_A}{2\Delta} & b_B \geq \frac{1}{\Delta}
		\end{cases}
	\end{align*}

	Again, notice that they coincide at $b_B = \frac{1}{\Delta}$, and increasing $b_B$ beyond $\frac{1}{\Delta}$ does not improve B's payoff. The first order condition on the interior part of the curve is:
	\begin{align*}
		\Delta v_B - (1-\delta_A \Delta) b_B - \Delta\delta_B v_B  = 0 &\implies \Delta b_B (1-\delta_A)= \Delta v_B - \Delta \delta_B v_B\\& \implies b_B = \Delta v_B.
	\end{align*}
	Again, $u_B(b_B)$ is strictly concave over $[0, \frac{1}{\Delta}]$, so this is a maximizer, and like $u_A$, $u_B$ is continuous with two pieces, and the strict concavity and cap guarantees that bidding $\Delta v_B$ gives at least as high payoff of bidding $\frac{1}{\Delta}$ or more. (Notice also that since $\Delta \leq 1$, the bidding strategy $b_B = \Delta v_B$ will never prescribe overbidding because $\Delta v_B \leq \frac{1}{\Delta}$.)
	
	Thus, $b_B = \Delta v_B$ is a best response to $b_A = v_A/\Delta$, and hence the pair is a Bayes-Nash equilibrium. 
\end{proof}


\begin{proof}[Proof of \Cref{thm:revgreedyvcg}]
	In the equilibrium described, we have that
	\begin{align*}
		\text{A wins} \iff  b_A \geq b_B \iff v_A \frac{1-\delta_A}{1-\delta_B} \geq v_B\frac{1-\delta_B}{1-\delta_A} \iff v_A \geq v_B \Delta^2.
	\end{align*}
	If $A$ wins, she pays $(1-\delta_B) b_B=(1-\delta_B) \Delta v_B$. If $B$ wins, he pays $(1-\delta_A)b_A={(1-\delta_A) \over \Delta} v_A$. 
	So revenue is given by:
		\begin{align*}
		\E[R(v_A,v_B)] = &\int_0^1 \int_0^{v_B\Delta^2} \frac{(1-\delta_A)}{\Delta} v_A dv_A dv_B + \int_0^1 \int_{\Delta^2 v_B}^1 (1-\delta_B) \Delta v_B dv_A dv_B \\
	\end{align*}

But notice that this is exactly the same equilibrium described in \Cref{eq:revoptgs} in  \Cref{prop:optgsprev}, and thus the calculation follows exactly the same way. 
\end{proof}

\begin{proof}[Proof of \Cref{thm:optvcgrev}] In this equilibrium and mechanism, A wins iff $v_A \geq \Delta v_B$, and pays $(1-\delta_B) v_B$. Otherwise, B wins and pays $(1-\delta_A) v_A$. Hence revenue is:
	\begin{align*}
		\E[R(v_A,v_B)]&= \int_{0}^{1} \int_{0}^{\Delta v_B} (1-\delta_A) v_A dv_A dv_B + \int_0^1 \int_{\Delta v_B}^{1} (1-\delta_B) v_B dv_A dv_B
	\end{align*}
But again, we notice that this is exactly the same equation as \Cref{eq:revoptgs} in \Cref{prop:gspgreedyrev}, so again, the calculation follows in exactly the same way.
\end{proof}

\begin{proof}[Proof of Theorem \ref{thm:eqrev}]
	The two equalities follow by inspection, so we only need to prove the inequality. 
	\begin{align*}
		R_{\greedy}^{*\gsp} - R_{\greedy}^{*\vcg} =& \frac{1-\delta_A}{6}\Delta^2 + \frac{1-\delta_B}{6} (3-2\Delta) \\&-\frac{1-\delta_A}{6} \Delta^3 - \frac{1-\delta_B}{6} \Delta(3-2\Delta^2)\\=&
		\frac{1-\delta_A}{6} (\Delta^2-\Delta^3) + \frac{1-\delta_B}{6} \left(3-2\Delta - 3\Delta + 2\Delta^3\right)
	\end{align*}
	Expanding :
	\begin{align*}
		&=\frac{1-\delta_A}{6} \Delta^2 - \Delta^3\frac{1-\delta_A}{6} + 2 \Delta^3\frac{1-\delta_B}{6}  -5\Delta \frac{1-\delta_B}{6}+\frac{3(1-\delta_B)}{6} \\&= \frac{1-\delta_A}{6} \Delta^2  + \Delta^3\left(\frac{1-\delta_A}{6} + \frac{2(1-\delta_B)}{6}\right)  -\frac{5\Delta(1-\delta_B)}{6} + \frac{3(1-\delta_B)}{6}\\
		&= \frac{ \Delta^2(1-\delta_A)}{6} +\Delta^3\left(\frac{3 - \delta_A - 2\delta_B}{6}\right) - \frac{5\Delta(1-\delta_B)}{6} + \frac{3(1-\delta_B)}{6}\\
		&= \frac{\Delta^2(1-\delta_A) + \Delta^3(3-\delta_A-2\delta_B) - 5\Delta(1-\delta_B)+3(1-\delta_B)}{6}
	\end{align*}
	Using $\delta_A \leq \delta_B\implies -\delta_A \geq -\delta_B$, we have:
	\begin{align*}
		&\frac{\Delta^2(1-\delta_A) + \Delta^3(3-\delta_A-2\delta_B) - 5\Delta(1-\delta_B) +3(1-\delta_B)}{6} \\&\geq \frac{\Delta^2(1-\delta_B) + \Delta^3(3-3\delta_B) -5\Delta(1-\delta_B)+3(1-\delta_B)}{6}\\
		&= \frac{1-\delta_B}{6} \left[\Delta^2 + 3\Delta^3 -5\Delta +3 \right]
	\end{align*}
	On the range $\Delta \in [0,1]$, the inner function is positive. To see this, one can either graph the function using a computer algebra system, or prove this analytically. For completeness: Note that $3+ \Delta^2 + 3\Delta^3 -5\Delta$ is bounded below by $3+ \Delta^2 + \Delta^3-5\Delta$. So it suffices to show that the latter is positive on $\Delta \in [0,1]$. So let $f(\Delta) = \Delta^2 + \Delta^3 - 5\Delta$. Then notice that $f(0) = 0$, $f(1) = -3$, and $f'(\Delta)$ is given by $3\Delta^2 + 2\Delta -5$. Since $\Delta <1$, $f'$ is always negative on $[0,1]$. But that means that, given that $f(0)=0$ and $f(1)=-3$, $f$ cannot go below $-3$ on the interval (otherwise it would have to have a positive derivative at some point to come back up to $-3$). 
	
	Hence, we conclude that $f(\Delta) \geq - 3 \ \forall \Delta \in [0,1]$, and so $3+ f(\Delta) \geq 0$. Tracing back through the inequalities, this gives $R_{\greedy}^{*\gsp} \geq R_{\greedy}^{*\vcg}$, and the claim follows. 
\end{proof}

%% file: main.bbl
\begin{thebibliography}{10}

\bibitem{AGV07}
{\sc Abrams, Z., Ghosh, A., and Vee, E.}
\newblock Cost of conciseness in sponsored search auctions.
\newblock In {\em Proc. of 3rd International Conference on Web and Internet
  Economics\/} (2007), pp.~326--334.

\bibitem{ausubel2006lonely}
{\sc Ausubel, L.~M., and Milgrom, P.}
\newblock The lovely but lonely vickrey auction.
\newblock In {\em Combinatorial Auctions, chapter 1\/} (2006), MIT Press.

\bibitem{caragiannis2015bounding}
{\sc Caragiannis, I., Kaklamanis, C., Kanellopoulos, P., Kyropoulou, M.,
  Lucier, B., Leme, R.~P., and Tardos, {\'E}.}
\newblock Bounding the inefficiency of outcomes in generalized second price
  auctions.
\newblock {\em Journal of Economic Theory 156\/} (2015), 343--388.

\bibitem{cavallo2018matching}
{\sc Cavallo, R., Sviridenko, M., and Wilkens, C.~A.}
\newblock Matching auctions for search and native ads.
\newblock In {\em Proceedings of the 2018 ACM Conference on Economics and
  Computation\/} (2018), pp.~663--680.

\bibitem{CW14}
{\sc Cavallo, R., and Wilkens, C.~A.}
\newblock Gsp with general independent click-through-rates.
\newblock In {\em Proc. of 10th International Conference on Web and Internet
  Economics\/} (2014), T.-Y. Liu, Q.~Qi, and Y.~Ye, Eds., pp.~400--416.

\bibitem{chawla2013auctions}
{\sc Chawla, S., and Hartline, J.~D.}
\newblock Auctions with unique equilibria.
\newblock In {\em Proceedings of the fourteenth ACM conference on Electronic
  commerce\/} (2013), pp.~181--196.

\bibitem{chen2009settling}
{\sc Chen, X., Deng, X., and Teng, S.-H.}
\newblock Settling the complexity of computing two-player nash equilibria.
\newblock {\em Journal of the ACM (JACM) 56}, 3 (2009), 1--57.

\bibitem{C71}
{\sc Clarke, E.~H.}
\newblock Multipart pricing of public goods.
\newblock {\em Public choice 11}, 1 (1971), 17--33.

\bibitem{colini2020envy}
{\sc Colini-Baldeschi, R., Leonardi, S., Schrijvers, O., and Sodomka, E.}
\newblock Envy, regret, and social welfare loss.
\newblock In {\em Proceedings of The Web Conference 2020\/} (2020),
  pp.~2913--2919.

\bibitem{adtypes}
{\sc Colini-Baldeschi, R., Mestre, J., Schrijvers, O., and Wilkens, C.~A.}
\newblock The ad types problem.
\newblock In {\em Web and Internet Economics: 16th International Conference,
  WINE 2020, Ljublana, Slovenia, December 17--20, 2017, Proceedings\/} (2020),
  Springer.

\bibitem{edelman2007internet}
{\sc Edelman, B., Ostrovsky, M., and Schwarz, M.}
\newblock Internet advertising and the generalized second-price auction:
  Selling billions of dollars worth of keywords.
\newblock {\em American economic review 97}, 1 (2007), 242--259.

\bibitem{feng2021}
{\sc Feng, Z., Guruganesh, G., Liaw, C., Mehta, A., and Sethi, A.}
\newblock Convergence analysis of no-regret bidding algorithms in repeated
  auctions.
\newblock In {\em The Thirty-Fifth AAAI Conference on Artificial Intelligence
  (AAAI-21)\/} (2021).

\bibitem{forges1993five}
{\sc Forges, F.}
\newblock Five legitimate definitions of correlated equilibrium in games with
  incomplete information.
\newblock {\em Theory and decision 35}, 3 (1993), 277--310.

\bibitem{gomes2014bayes}
{\sc Gomes, R., and Sweeney, K.}
\newblock Bayes--nash equilibria of the generalized second-price auction.
\newblock {\em Games and economic behavior 86\/} (2014), 421--437.

\bibitem{G73}
{\sc Groves, T.}
\newblock Incentives in teams.
\newblock {\em Econometrica: Journal of the Econometric Society\/} (1973),
  617--631.

\bibitem{hartline2015no}
{\sc Hartline, J., Syrgkanis, V., and Tardos, E.}
\newblock No-regret learning in bayesian games.
\newblock In {\em Advances in Neural Information Processing Systems\/} (2015),
  pp.~3061--3069.

\bibitem{kaplan2012asymmetric}
{\sc Kaplan, T.~R., and Zamir, S.}
\newblock Asymmetric first-price auctions with uniform distributions: analytic
  solutions to the general case.
\newblock {\em Economic Theory 50}, 2 (2012), 269--302.

\bibitem{kleinberg2011beyond}
{\sc Kleinberg, R.~D., Ligett, K., Piliouras, G., and Tardos, {\'E}.}
\newblock Beyond the nash equilibrium barrier.
\newblock In {\em ICS\/} (2011), pp.~125--140.

\bibitem{kuhn1955hungarian}
{\sc Kuhn, H.~W.}
\newblock The hungarian method for the assignment problem.
\newblock {\em Naval research logistics quarterly 2}, 1-2 (1955), 83--97.

\bibitem{leme2010pure}
{\sc Leme, R.~P., and Tardos, E.}
\newblock Pure and bayes-nash price of anarchy for generalized second price
  auction.
\newblock In {\em 2010 IEEE 51st Annual Symposium on Foundations of Computer
  Science\/} (2010), IEEE, pp.~735--744.

\bibitem{lucier2011gsp}
{\sc Lucier, B., and Paes~Leme, R.}
\newblock {GS0}p auctions with correlated types.
\newblock In {\em Proceedings of the 12th ACM conference on Electronic
  commerce\/} (2011), pp.~71--80.

\bibitem{munkres1957algorithms}
{\sc Munkres, J.}
\newblock Algorithms for the assignment and transportation problems.
\newblock {\em Journal of the society for industrial and applied mathematics
  5}, 1 (1957), 32--38.

\bibitem{roughgarden2015intrinsic}
{\sc Roughgarden, T.}
\newblock Intrinsic robustness of the price of anarchy.
\newblock {\em Journal of the ACM (JACM) 62}, 5 (2015), 1--42.

\bibitem{roughgarden2016twenty}
{\sc Roughgarden, T.}
\newblock {\em Twenty lectures on algorithmic game theory}.
\newblock Cambridge University Press, 2016.

\bibitem{roughgarden2017price}
{\sc Roughgarden, T., Syrgkanis, V., and Tardos, E.}
\newblock The price of anarchy in auctions.
\newblock {\em Journal of Artificial Intelligence Research 59\/} (2017),
  59--101.

\bibitem{syrgkanis2013composable}
{\sc Syrgkanis, V., and Tardos, E.}
\newblock Composable and efficient mechanisms.
\newblock In {\em Proceedings of the forty-fifth annual ACM symposium on Theory
  of computing\/} (2013), pp.~211--220.

\bibitem{varian2007position}
{\sc Varian, H.~R.}
\newblock Position auctions.
\newblock {\em international Journal of industrial Organization 25}, 6 (2007),
  1163--1178.

\bibitem{vickrey1961counterspeculation}
{\sc Vickrey, W.}
\newblock Counterspeculation, auctions, and competitive sealed tenders.
\newblock {\em The Journal of finance 16}, 1 (1961), 8--37.

\end{thebibliography}
